\documentclass{article}
\makeatletter
\def\input@path{{../}{../figs/}}
\makeatother

\PassOptionsToPackage{numbers, compress}{natbib}

\usepackage[preprint]{style}

\usepackage[utf8]{inputenc} % allow utf-8 input
\usepackage[T1]{fontenc}    % use 8-bit T1 fonts
\usepackage[colorlinks=true,linkcolor=black,anchorcolor=black,citecolor=black,filecolor=black,menucolor=black,runcolor=black,urlcolor=black]{hyperref}
\usepackage{url}            % simple URL typesetting
\usepackage{booktabs}       % professional-quality tables
\usepackage{amsfonts}       % blackboard math symbols
\usepackage{nicefrac}       % compact symbols for 1/2, etc.
\usepackage{microtype}      % microtypography
\usepackage{algorithm}
\usepackage[noend]{algpseudocode}

\usepackage{amsmath,amsthm,amssymb}
\usepackage{mathtools}
\usepackage{lipsum,graphicx,multicol}

\usepackage{natbib}
\usepackage{microtype}
\usepackage{graphicx}
\usepackage{caption}
\usepackage{subcaption}
\usepackage{float} % better placement of figures/tables
\usepackage{booktabs} % for professional tables
\usepackage{commath}
\usepackage{bbm}
\usepackage{xcolor}
\usepackage{xspace}

\newcommand{\R}{\mathbb{R}}

\newcommand{\pr}{\mathbb{P}}

\newcommand{\E}{\mathbb{E}}
\newcommand{\Var}{\mathrm{Var}}

\newcommand{\inv}{^{-1}}
\newcommand{\ones}[1]{\mathbf{1}_{#1}}
\newcommand{\msqrt}{^{\frac{1}{2}}}
\newcommand{\nsqrt}{^{-\frac{1}{2}}}
\newcommand{\diag}{\mathrm{diag}}

\newcommand{\noise}{\sigma^2}
\newcommand{\precision}{\sigma^{-2}}
\newcommand{\trace}{\mathrm{tr}}

\newcommand{\gain}{\mathrm{Gain}(\pi, \sigma^2)}
\newcommand{\risk}{\mathrm{R}}
\newcommand{\loss}{\mathrm{L}}
\newcommand{\betaLS}{\hat \beta_{\mathrm{LS}}}
\newcommand{\betaEC}{\hat \beta_{\mathrm{ECov}}}
\newcommand{\betaMM}{\hat \beta_{\mathrm{ECov}}^{\mathrm{MM}}}
\newcommand{\betaED}{\hat \beta_{\mathrm{EData}}}
\newcommand{\betaEDMM}{\hat \beta_{\mathrm{EData}}^{\mathrm{MM}}}
\newcommand{\betaID}{\hat \beta_{\mathrm{ID}}}
\newcommand{\data}{\mathcal{D}}
\newcommand{\SB}{\Sigma}
\newcommand{\SBhat}{\hat \Sigma}
\newcommand{\SBMM}{\hat \Sigma^{\mathrm{MM}}}
\newcommand{\SBTrue}{\tilde \Sigma}
\newcommand{\GammaMM}{\hat \Gamma^{\mathrm{MM}}}

\newcommand{\glow}{g_{\lambda_{\mathrm{max}}}}
\newcommand{\vghigh}{\vec g_{\lambda_{\mathrm{min}}}}
\newcommand{\vglow}{\vec g_{\lambda_{\mathrm{max}}}}
\newcommand{\lambdaMin}{\lambda_{\mathrm{min}}}
\newcommand{\lambdaMax}{\lambda_{\mathrm{max}}}

\newcommand\mydots{\hbox to 1em{.\hss.\hss.}}

\usepackage[capitalize]{cleveref}

\def\[#1\]{\begin{align}#1\end{align}}        % numbered
\def\(#1\){\begin{align*}#1\end{align*}}      % unnumbered

\newtheorem{theorem}{Theorem}[section]
\newtheorem{corollary}[theorem]{Corollary}
\newtheorem{prop}[theorem]{Proposition}
\newtheorem{condition}[theorem]{Condition}
\newtheorem{definition}[theorem]{Definition}

\newtheorem{lemma}[theorem]{Lemma}
\theoremstyle{remark}
\newtheorem{remark}[theorem]{Remark}

\DeclareMathOperator*{\argmax}{arg\,max}
\DeclareMathOperator*{\argmin}{arg\,min}

\graphicspath{{../}} %Setting the graphicspath

\title{For high-dimensional hierarchical models, consider exchangeability of effects across covariates instead of across datasets}
\author{%
  Brian L.~Trippe \\
    MIT CSAIL \\
  \texttt{btrippe@mit.edu} \\
   \And
  Hilary K.~Finucane\\
  Broad Institute \\
  \texttt{finucane@broadinstitute.org}\\
   \AND
  Tamara Broderick \\
  MIT CSAIL \\
  \texttt{tbroderick@csail.mit.edu}\\
}

\begin{document}

\maketitle

\begin{abstract}
Hierarchical Bayesian methods enable information sharing across multiple related regression problems.
While standard practice is to model regression parameters (effects) as (1) exchangeable across datasets and (2) correlated to differing degrees across covariates,
we show that this approach exhibits poor statistical performance when the number of covariates exceeds the number of datasets.
For instance, in statistical genetics, we might regress dozens of traits (defining datasets) for thousands of individuals (responses) on up to millions of genetic variants (covariates).
When an analyst has more covariates than datasets, we argue that it is often more natural to instead model effects as (1) exchangeable across \emph{covariates} and (2) correlated to differing degrees across \emph{datasets}.
To this end, we propose a hierarchical model expressing our alternative perspective.
We devise an empirical Bayes estimator for learning the degree of correlation between datasets.
We develop theory that demonstrates that our method outperforms the classic approach when the number of covariates dominates the number of datasets,
and corroborate this result empirically on several high-dimensional multiple regression and classification problems.

\end{abstract}

%%%% Introduction \label{sec:intro}
\section{Introduction}\label{sec:intro}

Hierarchical modeling is a mainstay of Bayesian inference.
For instance, in (generalized) linear models, the unknown parameters are \emph{effects}, each of which describes the association of a particular covariate with a response of interest.
Often covariates are shared across multiple related datasets, but the effects are typically allowed to vary both by dataset and by covariate.
A classic methodology, dating back to \citeauthor{lindley1972bayes} (\citeyear{lindley1972bayes}) \citep{lindley1972bayes}, models the effects as conditionally independent across datasets, with a latent (and learnable) degree of relatedness across covariates.
From a practical standpoint, the model is motivated by the understanding that it ``borrows strength'' across different datasets \citep[Chapter 5.6]{gelman2013bayesian}.
Mathematically, the model is motivated by assuming effects are exchangeable across datasets and applying a de Finetti theorem \citep{lindley1972bayes,jordan2010bayesian}.
The methodology of \citeauthor{lindley1972bayes} is ubiquitous when the number of datasets is larger than the number of covariates.
It is a standard component of Bayesian pedagogy [\citep[Chapter 13.3]{gelman2006data}; \citep[Chapter 15.4]{gelman2013bayesian}] and software;
e.g. it is used by default in the mixed modeling package \texttt{lme4} \citep{bates2015fitting}, which has over 13 million downloads at the time of writing.

Despite its resounding success when there are more datasets than covariates, we show in the present work that the approach of \citeauthor{lindley1972bayes} performs poorly when there are more covariates than datasets.
To address the many-covariates case, we turn for inspiration to statistical genetics, where scientists commonly learn linear models relating genetic variants (covariates) to traits (corresponding to different datasets) across individuals (which each exhibit a response).
These applications may exhibit millions of covariates, thousands of responses, and just a handful of datasets.
In these cases, \citep{lee2012estimation,bulik_finucane2015atlas,stephens2013unified,zhou2014efficient,maier2015joint,runcie2020megalmm} use a multivariate Gaussian prior akin to that of \citeauthor{lindley1972bayes}, but instead assume conditional independence across \emph{covariates} and prior parameters that encode correlations across \emph{datasets}. 

We can imagine using a similar model in applications beyond statistical genetics.
Namely, when there are more covariates than datasets, we propose to model the effects as exchangeable across \emph{covariates} (rather than datasets) and learn the degree of relatedness of effects across \emph{datasets} (rather than covariates).
Henceforth we refer to this framework as \emph{ECov}, for \emph{e}xchangeable effects across \emph{cov}ariates, and distinguish it from \emph{e}xchangeable effects across \emph{data}sets, or \emph{EData}.

While the existing methods in statistical genetics for modeling multiple traits obtain as a special case of ECov, 
to the best of our knowledge this approach is absent from existing literature on hierarchical Bayesian regression.
\citeauthor{brown1980adaptive} (\citeyear{brown1980adaptive}) \citep{brown1980adaptive} and \citeauthor{haitovsky1987multivariate} (\citeyear{haitovsky1987multivariate}) \citep{haitovsky1987multivariate} form two exceptions, but these two papers (1) consider only the situation in which a single covariate matrix is shared across all datasets (or equivalently, for each data point all responses are observed) and (2) include only theory and no empirics.

We suspect that the historical origins of the methodology in statistical genetics may have hindered earlier expansion of this class of models to a wider audience.
In particular, this literature traces back to mixed effects modeling for cattle breeding \citep{thompson1973estimation};
here, an even-earlier notion of the genetic contribution of trait correlation (i.e.\ ``genetic correlation;'' see \citeauthor{hazel1943genetic} (\citeyear{hazel1943genetic}) \citep{hazel1943genetic}) informs the covariance structure of random effects.
Although genetic correlation is now commonly understood to describe the correlation of effects of DNA sequence changes on different traits \citep{bulik_finucane2015atlas}, its provenance predates even the first identification of DNA as the genetic material in \citeyear{avery1944studies} \citep{avery1944studies}.
As such, this older motivation obviated the need for a more general justification grounded in exchangeability.
See \Cref{sec:additional_related_work} for further discussion of related work.

In the present work, we propose ECov as a general framework for hierarchical regression when the number of covariates exceeds the number of datasets.
We show that the classic model structure from statistical genetics can be seen as an instance of this framework, much as \citeauthor{lindley1972bayes} give a (complementary) instance of an EData framework.
To make the ECov approach generally practical, we devise an accurate and efficient algorithm for learning the dataset-correlation matrix.
We demonstrate with theory and empirics that ECov is preferred when the number of covariates exceeds the number of datasets, while EData is preferred when the number of datasets exceeds the number of covariates.
Our experiments analyze three real, non-genetic datasets in regression and classification, including an application to transfer learning with pre-trained neural network embeddings.
We provide proofs of theoretical results in the appendix.

%%%% Exchangeability among covariate effects \label{sec:exch_cov}
\section{Exchangeability and its applications to hierarchical linear modeling}\label{sec:exch_cov}

We start by establishing the data and model, motivating exchangeability among covariate effects (ECov), and motivating our Bayesian generative model.

\textbf{Setup and notation.} Consider $Q$ datasets with $D$ covariates.
Let $N^q$ be the number of data points in dataset $q$.
For the $q$th dataset, the $N^q \times D$ real design matrix $X^q$ collects the covariates, and $Y^q$ is the $N^q$-vector of responses.
The $n$th datapoint in dataset $q$ consists of covariate $D$-vector $X^q_n$ and scalar response $Y^q_n$. 
We let $\data := \{(X^q,Y^q)\}_{q=1}^Q$ denote the collection of all $Q$ datasets.
We consider the generalized linear model $Y^q_n | X^q_n, \beta^q \overset{indep}{\sim} p(\cdot | X^{q\top}_n \beta^q)$ with unknown $D$-vector of real effects $\beta^q$.
We collect all effects in a $D\times Q$ matrix $\beta$ with $(d,q)$ entry $\beta^q_d$.
The linear form of the likelihood allows interpretation of $\beta^q_d$ as the association between the $d$th covariate and the response in dataset $q$.
In linear regression, the responses are real-valued and the conditional distribution is Gaussian.
In logistic regression, the responses are binary, and we use the logit link.
The independence assumption conflicts with some models that one might use, for example in some cases when the different datasets partially overlap.

\textbf{Example.} As a motivating non-genetics example, consider a study of the efficacy of microcredit.
There are seven famous randomized controlled trials of microcredit, each in a different country \citep{meager2019understanding}.
We might be interested in the association between various aspects of small businesses (covariates), including whether or not they received microcredit, and their business profit (response).
In this case, the $d$th element of $X^q_n$ would be the $d$th characteristic of the $n$th small business in the $q$th country, and $Y^q_n$ is the profit of this business.
See the experiments for additional examples in rates of policing, web analytics, and transfer learning.

\textbf{Exchangeable effects across datasets (EData).} To fully specify a Bayesian model, we need to choose a prior over the parameters $\beta$.
\citeauthor{lindley1972bayes} assume the effects are exchangeable across datasets.
Namely, for every $Q$-permutation $\sigma$, $p(\beta^1, \beta^2, \dots, \beta^Q) = 
    p(\beta^{\sigma(1)}, \beta^{\sigma(2)}, \dots, \beta^{\sigma(Q)})$.
Assuming exchangeability holds for an imagined growing $Q$ and applying de Finetti's theorem motivates a conditionally independent prior.
Concretely, \citeauthor{lindley1972bayes} take $\beta^q \overset{i.i.d.}{\sim} \mathcal{N}(\xi, \Gamma)$, for $D$-vector $\xi$ and $D\times D$ covariance matrix $\Gamma$.
The $(d,d')$ entry of $\Gamma$ captures the degree of relatedness between the effects for covariates $d$ and $d'$.
Both $\xi$ and $\Gamma$ may be learned in an empirical Bayes procedure.

\textbf{Exchangeable effects across covariates (ECov).} We here argue for a complementary approach in settings where $D>Q$.
In the microcredit example, 
notice that $D>Q$ will arise whenever the experimenter records more characteristics of a small business than there are locations with microcredit experiments; that is, $D>7$ in this particular case.
Concretely, let $\beta_d$ be the $Q$-vector of effects for covariate $d$ across datasets.
Then, in the ECov approach, we will assume that effects are exchangeable across covariates instead of across datasets.
Namely, for every $D$-permutation $\sigma$, $p(\beta_1, \beta_2, \dots, \beta_D) = 
    p(\beta_{\sigma(1)}, \beta_{\sigma(2)}, \dots, \beta_{\sigma(D)})$.
We will see theoretical and empirical benefits to ECov in later sections, but note that the ECov assumption is often a priori natural.
For instance, regarding microcredit, we may have no a priori beliefs about how effects differ for distinct small-business characteristics.
And we may a priori believe that different countries could exhibit more similar effects -- and wish to learn the degree of relatedness across those countries.

We may apply a similar rationale as \citeauthor{lindley1972bayes} to motivate a conditionally independent model.
Analogous to \citeauthor{lindley1972bayes}, we propose a Gaussian prior: $\beta_d \overset{i.i.d.}{\sim} \mathcal{N}(0, \SB)$. $\SB$ is now a $Q\times Q$ covariance matrix
whose $(q,q')$ entry captures the similarity between the effects in the $q$ and $q'$ datasets.
For simplicity, we restrict to $\E[\beta_d]=0$; see \Cref{sec:non_zero_means} for discussion.
Another potential benefit to ECov relative to EData is that we might expect a statistically easier problem, with $O(Q^2)$ rather than $O(D^2)$ values to learn in the relatedness matrix.
We provide a rigorous theoretical analysis in \Cref{sec:nonasymptotic_theory,sec:asymptotic_gains}.

%%%% Our Method (Comparison to exchangeability among regressions, two empirical Bayes estimates
%%%% and EM algorithm  \label{sec:our_method}
\section{Our method}\label{sec:our_method}

We next describe our inference method for specific instances of the exchangeable covariate effects model of \Cref{sec:exch_cov}.
We compute the $\beta$ posterior and take an empirical Bayes approach to estimate $\SB$.
We find that an expectation maximization (EM) algorithm estimates $\SB$ effectively;
\Cref{sec:inference_of_gamma_approaches} compares our approach to existing methods for the related problem of estimating $\Gamma$ for EData.

\textbf{Notation.} We identify estimates of $\beta$ and $\SB$ with hats.
For instance, $\betaLS$ is the least squares estimate, with $\betaLS^q := (X^{q\top} X^q)\inv X^{q\top}Y^q.$ 
We will sometimes find it useful to stack the columns of $\beta$ or its estimates into a length $DQ$ vector;
we denote such vectors with an arrow;
for example, $\vec \beta := [\beta^{1\top}, \beta^{2\top},\dots, \beta^{Q\top}]^\top.$
For a natural number $N,$ we use $I_N, \ones{N},$ and $e_N$ to denote the $N\times N$ identity matrix, $N$-vector of ones, and $N$th basis vector, respectively.
We use $\otimes$ to denote the Kronecker product.

\subsection{Posterior inference with a Gaussian likelihood}\label{sec:conjugate_gaussian_estimate}
We first consider a Gaussian likelihood:
for each dataset $q$ and observation $n$, we take $Y^q_n |X_n^q, \beta^q  \overset{indep}{\sim} \mathcal{N}(X^{q\top}_n \beta^q, \sigma_q^2)$
where $\sigma_q^2$ is a dataset-specific variance.
When the relatedness matrix $\SB$ is known, a natural estimate of $\beta$ is its posterior mean.
We obtain the full posterior, including its mean, via a standard conjugacy argument; see \Cref{sec:proof_of_conj_form}:
\begin{prop}\label{prop:conj_form}
For each covariate $d$, let
 $\beta_d \overset{i.i.d.}{\sim} \mathcal{N}(0, \SB)$ a priori.
For each dataset $q$ and data point $n,$ let $Y^q_n|X_n^q,\beta^q \overset{indep}{\sim}\mathcal{N}(X^{q\top}_n\beta^q,\sigma_q^2).$
Then $\vec \beta | \data, \SB \sim \mathcal{N}(\vec \mu, V)$ for 
$\vec \mu = 
V [\sigma_1^{-2}Y^{1\top}X^1, \dots, \sigma_Q^{-2}Y^{Q\top}X^Q]^\top$
and
$V\inv= \SB\inv \otimes I_D +  \diag(\sigma_1^{-2}X^{1\top}X^1, \dots,\sigma_Q^{-2}X^{Q\top}X^Q),$
where $\diag(\sigma_1^{-2}X^{1\top}X^1, \dots,\sigma_Q^{-2}X^{Q\top}X^Q)$ denotes a $DQ\times DQ$ block-diagonal matrix.
\end{prop}

At first glance, the posterior mean $\vec \mu$ for this model 
might seem to introduce a computational challenge because exact computation of $V$ involves an $O(D^3Q^3)$-time matrix inversion.
Our experiments (\Cref{sec:experiments}), however, involve on the order of $DQ\approx 1{,}000$ parameters, so
direct inversion of $V$ demands less than a single second.
Moreover, in much larger problems $\vec \mu$ may still be computed very efficiently using the conjugate gradient algorithm \citep[Chapter 5]{nocedal2006numerical},
with convergence in a small number of $O(D^2Q)$ time iterations; see \Cref{sec:conj_gradients}.
%Nevertheless, we suspect that the comparatively large computational hassle is one reason why this approach has yet been under-explored.

\subsection{Empirical Bayes estimation of $\SB$ by expectation maximization}\label{sec:empirical_bayes_and_EM}
The posterior mean of $\beta$ in \Cref{prop:conj_form} requires $\SB,$
which is typically unknown.
Accordingly, we propose an empirical Bayes approach of estimating $\SB$ by maximum marginal likelihood:
\[\label{eqn:empirical_bayes_estimate}
\betaEC := \E[\beta \mid \data, \SBhat]
\text{ where }\SBhat := \argmax_{\SB\succeq 0} p(\data \mid \SB).
\]
\Cref{eqn:empirical_bayes_estimate} defines a two step procedure.
In the first step, we learn the similarity between datasets via estimation of $\SB.$
In the second step, we use this similarity to compute an estimate, $\betaEC,$ that correspondingly shares strength.
Though we have been unable to identify a general analytic form for $\SBhat,$ we can compute it with an expectation maximization (EM) algorithm \citep[Chapter 1.5]{mclachlan2007algorithm}.
\Cref{alg:EM_general} summarizes this procedure; see \Cref{sec:em_details} for details.
%%As we show in \Cref{sec:experiments}, one can independently estimate the variances $\noise_q,$ when these are unknown.

\begin{figure}[!t]
\begin{minipage}{0.50\textwidth}
\begin{algorithm}[H]
%\centering
\caption{
    Expectation Maximization for Exchangeability Among Covariate Effects
}
\label{alg:EM_general}
\begin{algorithmic}[1]
% %\LinesNumbered
\State \text{// Initialize covariance}
\State $\SB^{(0)} \leftarrow I_Q$
\State \text{// Run EM algorithm}
\For{$t = 0,1,\dots$}
    \State \text{// Expectation step}
    \State $\mu_1, \mydots, \mu_D, V_1, \mydots, V_D \leftarrow\texttt{E\_Step}(\SB^{(t)})$
    \State \text{ }
    \State \text{// Maximization step}
    \State $\SB^{(t+1)}  \leftarrow D\inv \sum_{d=1}^D (\mu_d \mu_d^\top + V_d)$
\EndFor
\State \text{ }
\State \text{Return} $\SB^{(t+1)}$
\end{algorithmic}
\end{algorithm}
\end{minipage}
\hfill
\begin{minipage}{0.46\textwidth}
\begin{algorithm}[H]
%\centering
\caption{
E-Step: Linear Regression
}
\label{alg:E_step_conj}
\begin{algorithmic}[1]
    \State $\vec \mu, V  \leftarrow \E[\vec \beta | \data, \SB], \Var[\vec \beta | \data, \SB]$ 
    \For{$d = 1,\mydots, D$}
        \State $\mu_d \leftarrow (e_d \otimes I_Q)^\top \vec \mu$ 
        \State $V_d \leftarrow (e_d \otimes I_Q)^\top V (e_d \otimes I_Q)$ 
    \EndFor
    \State \text{Return } $\mu_1, \mydots, \mu_D, V_1, \mydots, V_D$
\end{algorithmic}
\end{algorithm}
\vspace{-15pt}
\begin{algorithm}[H]
%\centering
\caption{
E-Step: Logistic Regression
}
\label{alg:E_step_logistic}
\begin{algorithmic}[1]
 \State $\vec \mu^* \leftarrow \argmax_{\vec\beta} \log p(\vec\beta | \data, \SB)$ 
 \State $V \leftarrow -[\nabla_\beta^2 \log p(\vec\beta | \data, \SB)\big|_{\vec \beta = \vec \mu^*}]\inv$ 
 \For{$d = 1,\mydots, D$}
    \State $\mu_d \leftarrow (e_d \otimes I_Q)^\top \vec \mu^*$ 
    \State $V_d \leftarrow (e_d \otimes I_Q)^\top V (e_d \otimes I_Q)$ 
\EndFor
\State \text{Return } $\mu_1, \mydots, \mu_D, V_1, \mydots, V_D$
\end{algorithmic}
\end{algorithm}
\end{minipage}
\end{figure}

\subsection{Classification with logistic regression}\label{sec:logistic_regression_empirical_bayes_main}
We can extend the approach above to inference for multiple related classification problems.
We assume a logistic likelihood; for each $q$ and $n$,
$Y^q_n | X_n^q, \beta^q \overset{indep}{\sim} \text{Bern}[(1+\exp\{-X_n^{q\top} \beta^q\})\inv].$
In the classification case, we cannot use Gaussian conjugacy directly, so we apply an approximation.
Specifically, we adapt the original E-step in \Cref{alg:E_step_logistic} by using
a Laplace approximation to the posterior \citep[Chapter 4.4]{Bishop2006}.
We approximate the posterior mean of $\beta$ by the maximum a posteriori value.
We leave extensions to other generalized linear models to future work.

%%%% Non-asymptotic Risk, Domination results and improvement of MLE
\section{Theoretical comparison of frequentist risk}\label{sec:nonasymptotic_theory}

In this section, we prove theory that suggests ECov has better frequentist risk than EData when $D$ is large relative to $Q$.
Analyzing $\betaEC$ directly is challenging due to its non-differentiability as a function of the data, so we take a multipart approach.
First, in \cref{theorem:exch_reg_domination}, we show that an ECov estimate based on moment-matching (MM), $\betaMM$, dominates least squares, $\betaLS$, when $D$ is large relative to $Q$; $\betaLS$ in turn dominates $\betaEDMM$ (a similar estimator for EData).
Second, in \cref{theorem:pos_part_dominance}, we show that $\betaEC$ uniformly improves on $\betaMM$.

\textbf{Setup.}
Take a fixed value of $\beta$ and an estimator $\hat \beta$. We use squared error risk,
$\risk(\beta, \hat \beta) := \E\left[ \|\hat \beta - \beta\|_F^2 \mid \beta\right]$,
as our measure of performance. $\|\cdot\|_F$ is the Frobenius norm of a matrix,
and the expectation is over all observations $Y^1, \dots, Y^Q$ jointly.
We require the following orthogonal design condition.
\begin{condition}\label{condition:orthogonal_design}
    For each dataset $q$, $\sigma_q^{-2} X^{q\top}X^q = \precision I_D$ for some shared variance $\noise$.
\end{condition}
Though restrictive, this condition is useful for theory, as other authors have found; see \Cref{sec:identity_covariance_condition}.
We empirically demonstrate that our theoretical conclusions apply more broadly in \Cref{sec:experiments}.

\textbf{ECov vs.\ EData when using moment matching in high dimensions.}
For ECov, the following estimate for $\SB$ is unbiased under correct prior specification:
$\SBMM := D\inv \betaLS^\top \betaLS 
- D\inv\diag(\sigma_1^2 \|X^{1\dagger}\|_F^2, \dots, \sigma_Q^2 \|X^{Q\dagger}\|_F^2),$
where $\dagger$ denotes the Moore-Penrose pseudoinverse of a matrix and $\betaLS$ is the least squares estimate.
We define $\betaMM := \E[\beta |\data, \SBMM]$ to be the resulting parameter estimate,
and define $\betaEDMM$ analogously for EData; see \Cref{sec:moment_ests_supp} for details.
While $\betaMM$ and $\betaEDMM$ are naturally defined only when $D\ge Q$ and $D\le Q,$ respectively,
we find it informative to compare how their performances depend on $D$ and $Q$ nonetheless.

Before our theorem, a lemma provides concise expressions for the risks of $\betaMM$ and $\betaEDMM.$
\begin{lemma}\label{lemma:form_of_risks}
Under \Cref{condition:orthogonal_design}
and when $D\ge Q,$
$\risk(\beta, \betaMM) =
\noise DQ - \sigma^4 D(D-2-2Q) \E[\|\betaLS^\dagger\|_F^2 \mid \beta].$
Additionally, when $D\le Q,$
$\risk(\beta, \betaEDMM) =
\noise DQ - \sigma^4 Q(Q-2-2D) \E[\|\betaLS^\dagger\|_F^2 \mid \beta].$
\end{lemma}
\Cref{lemma:form_of_risks} reveals forms for the risks of $\betaMM$ and $\betaEDMM$ that are surprisingly simple.
The symmetry between the forms and risks of these estimators, however,
is intuitive; under \Cref{condition:orthogonal_design},
$\betaMM$ and $\betaEDMM$ can be seen as respectively arising from the same procedure applied to $\betaLS$ and its transpose.

With \Cref{lemma:form_of_risks} in hand, we can now compare the risk of $\betaMM$, $\betaLS$, and $\betaEDMM$.
\begin{theorem}\label{theorem:exch_reg_domination}
Let \Cref{condition:orthogonal_design} hold.
Then 
(1) if $D>2Q+2,$ $\betaMM$ dominates $\betaLS$ with respect to squared error risk. 
In particular, for any $\beta,\  
\risk(\beta, \betaMM ) < 
\risk(\beta,  \betaLS ).$
Additionally, (2) if $D> Q/2 -1,$
$\betaEDMM$ is dominated by $\betaLS.$
\end{theorem}

Since $\betaLS$ is minimax \citep[Chapter 5]{lehmann2006theory}, \Cref{theorem:exch_reg_domination} implies that $\betaMM$ has minimax risk in the high-dimensional setting.
It follows that, regardless of how well the ECov prior assumptions hold, $\betaMM$ will not perform very poorly.
%\footnote{
%By a similar line of reasoning we can see that the $\betaEDMM$ does well in regime of more datasets than covariates.}

%%% Shrinkage perspective and positive part estimate
\textbf{Further improvement with maximum marginal likelihood.}
The moment based approach analyzed above has a limitation:
with positive probability, $\SBMM$ is not positive semi-definite (PSD).
Though our expression for $\betaMM$ remains well-defined in this case, this non-positive definiteness obscures the interpretation of $\betaMM$ as a Bayes estimate.
We next show that performance further improves if $\SB$ is instead estimated by maximum marginal likelihood (\Cref{eqn:empirical_bayes_estimate}) and is thereby constrained to be PSD.

Our next lemma characterizes the form of the resulting estimator, $\betaEC,$
and establishes a connection to the positive part James-Stein estimator \citep{baranchik1964multiple}.
\begin{lemma}\label{lemma:form_of_mle_est}
Assume $D>Q$ and consider the singular value decomposition $\betaLS = V \diag(\lambda\msqrt) U^\top$ 
where $V$ and $U$ satisfy $V^\top V= U^\top U=  I_Q,$
and $\lambda$ is a $Q$-vector of non-negative reals.
Under \Cref{condition:orthogonal_design},
\Cref{eqn:empirical_bayes_estimate} reduces to
$\SBhat= U\diag\left[ 
(D\inv\lambda - \noise \ones{Q})_+
\right]U^\top $
and 
$
\betaEC = V\diag\left[
    \lambda\msqrt\odot(\ones{Q} - \noise D\lambda\inv)_+
\right]U^\top,
$
where $(\cdot)_+$ is shorthand for $\max(\cdot, 0)$ element-wise,
$\odot$ is the Hadamard (i.e. element-wise) product, 
and the powers in $\lambda\msqrt$ and $\lambda\inv$ are applied element-wise.
\end{lemma}
\Cref{lemma:form_of_mle_est} allows us to see $\betaEC$ as shrinking $\betaLS$ toward $0$ in the direction of each singular vector to an extent proportional to the inverse of the associated singular value.
The transition from $\betaMM$ to $\betaEC$ is then analogous to the taking the ``positive part'' of the James-Stein estimator in vector estimation,
which provides a uniform improvement in risk \citep{baranchik1964multiple}.
Though $\risk(\beta,\betaEC)$ is not easily available analytically,
we nevertheless find that it dominates its moment-based counterpart.

\begin{theorem}\label{theorem:pos_part_dominance}
Assume $D>Q+1.$ 
Under \Cref{condition:orthogonal_design} $\betaEC$ dominates $\betaMM$ with respect to squared error loss, achieving strictly lower risk for every value of $\beta$.
\end{theorem}
We establish \Cref{theorem:pos_part_dominance} using a proof technique adapted from \citet{baranchik1964multiple}; see also \citet{lehmann2006theory}[Thm. 5.5.4].
The standard approach we build upon is complicated by the fact that the directions in which we apply shrinkage are themselves random.

\Cref{theorem:pos_part_dominance} provides a strong line of support for using $\betaEC$ over $\betaMM$ that does not rely on any assumption of ``correct'' prior specification;
in particular the risk improves without any subjective assumptions on $\beta.$
We discuss related earlier work in \Cref{sec:related_work_normal_means}.
 % \label{sec:nonasymptotic_theory}

%%%% Asymptotic for Gains Relative to Independent Analyses
\section{Gains from \textrm{ECov} in the high-dimensional limit}\label{sec:asymptotic_gains}
The results of \Cref{sec:nonasymptotic_theory} give a promising endorsement of ECov but face two important limitations.
First, the domination results relative to least squares do not directly demonstrate that $\betaEC$ attains improvements by leveraging similarities across datasets in a meaningful way;
indeed for a single dataset (i.e.\ $Q=1$) $\betaEC$ can be understood as a ridge regression estimate \citep{hoerl1970ridge}, 
and \Cref{theorem:exch_reg_domination,theorem:pos_part_dominance} provide that $\betaEC$ dominates $\betaLS$ for $D>3.$
Second, domination results reveal nothing about the size of the improvement or how it depends on any structure of $\beta$;
intuitively, we should expect better performance when $\beta$ is in some way representative of the assumed prior.
To address these limitations, we analyze the size of the gap between the risk of (1) $\betaEC$ and (2) our method applied to each dataset independently (ID),
which we denote by $\betaID$.
We will characterize the dependence of this gap on $\beta.$

Reasoning quantitatively about the dependence of the risk on the unknown parameter poses significant analytical challenges.
In particular, \Cref{lemma:form_of_risks} shows that $\risk(\beta, \betaMM)$ depends on $\beta$ through $\E[\| \betaLS^\dagger\|_F^2| \beta];$
however, $\| \betaLS^\dagger\|_F^2$ is the sum of the eigenvalues of a non-central inverse Wishart matrix,
a notoriously challenging quantity to work with; see e.g. \citep{letac2004tutorial,hillier2019properties}.
To regain tractability, we 
(1) develop an analysis asymptotic in the number of covariates $D$ and
(2) shift to a Bayesian analysis in order to sensibly consider a growing collection of covariate effects.
In particular, we consider a sequence of regression problems, with parameters $\{\beta_d\}_{d=1}^\infty$ distributed as
$\beta_d\overset{i.i.d.}{\sim}\pi$ for some distribution $\pi.$
Accordingly, instead of using the frequentist risk as in \Cref{sec:nonasymptotic_theory}, we now use the Bayes risk to measure performance.
Specifically, for a dataset with $D$ covariates and an estimator $\hat \beta,$ the Bayes risk is $\risk_\pi^D(\hat \beta) := \E_\pi[ \risk(\beta, \hat \beta)]$ where $\risk(\beta, \hat \beta)$ is the usual frequentist risk.

For a single metric characterizing the benefits of joint modeling, we will define the \emph{asymptotic gain} as the relative performance between our two estimators of interest here, $\betaEC$ and $\betaID$.
\begin{definition}\label{definition:gain}
Consider a sequence of datasets of $Q$ regression problems with an increasing number of covariates $D, \{\data_D\}_{D=1}^\infty.$
Assume that for each dataset \Cref{condition:orthogonal_design} holds with variance $\noise$ and that each $\beta_d\overset{i.i.d.}{\sim} \pi.$
The asymptotic gain of joint modeling is
$\gain := \lim_{D\rightarrow\infty} (\noise DQ)\inv[ \risk^D_\pi(\betaID) -\risk^D_\pi(\betaEC)].$
\end{definition}
The factor of $\noise DQ$ in \Cref{definition:gain} puts $\gain$ on a scale that is roughly invariant to the size and noise level of the problem;
for example, $(\noise DQ)\inv \risk^D_\pi(\betaLS)=1$ for any $\pi, D,$ and $Q.$

Our next lemma gives an analytic expression for $\gain$ that provides a starting point for understanding its problem dependence.
\begin{lemma}\label{lemma:gain}
Assume $\SBTrue := \Var_\pi[\beta_1]$ is finite and has eigenvalues $\lambda_1, \dots, \lambda_Q.$
Under \Cref{condition:orthogonal_design},
$\gain =  \noise Q\inv[\sum_{q=1}^Q (\lambda_{q} + \noise)\inv  - \sum_{q=1}^Q (\SBTrue_{q,q} + \noise)\inv]
.$
\end{lemma}
\Cref{lemma:gain} reveals that the diagonals and eigenvalues and $\SBTrue$ are key determinants of $\gain,$
but does not directly provide an interpretation of when $\betaEC$ offers benefits over $\betaID$.
Our next theorem demonstrates when an improvement can be achieved from joint modeling.
\begin{theorem}\label{theorem:asymptotic_gain}
$\gain \ge 0,$
with equality only when $\SBTrue = \Var_\pi[\beta_1]$ is diagonal.
\end{theorem}
\begin{proof}
From \Cref{lemma:gain} we see $\gain$ is the difference between a strictly Schur-convex function applied to the eigenvalues of $\SBTrue$ and to its diagonals (since  $(x+ \noise)\inv$ is convex on $\R_+$).
By the Schur-Horn theorem, the eigenvalues of $\SBTrue$ majorize its diagonals, providing the result.
\end{proof}
\Cref{theorem:asymptotic_gain} tells us that $\betaEC$ succeeds at adaptively learning and leveraging similarities among datasets in the high-dimensional limit.
In particular, $\gain$ reduces to zero only when the eigenvalues of $\SBTrue$ are arbitrarily close to the entries of its diagonal,
which occurs only when the covariate effects are uncorrelated across datasets.
However, when covariate effects are correlated, we obtain an improvement.

Our next theorem quantifies this relationship through upper and lower bounds.
\begin{theorem}\label{thm:gain_bounds}
Let $\lambda^\downarrow$ and $\ell^\downarrow$ denote the eigenvalues and diagonals of $\SBTrue$, respectively, sorted in descending order.
Then $\gain \le  2\noise Q\inv \|\lambda\|_2 \| \ell^\downarrow - \lambda^\downarrow \|_2/(\lambdaMin+\noise)^3$
and
$\gain \ge \noise Q\inv \| \ell^\downarrow - \lambda^\downarrow \|_2^2/(\lambdaMax + \noise)^{3},$
where $\lambdaMax$ and  $\lambdaMin$ are the largest and smallest, respectively, eigenvalues of $\SBTrue$.
\end{theorem}
\Cref{thm:gain_bounds} allows us to see several aspects of when our method will and will not perform well.
First, the presence of $\|\ell^\downarrow - \lambda^\downarrow \|_2^2$ in both the upper and lower bounds demonstrates that
$\gain$ will be small when the eigenvalues are close to the diagonal entries, with Euclidean distance as an informative metric.

As we find in our next corollary, \Cref{thm:gain_bounds} additionally allows us to see that nontrivial gains may be obtained only in an intermediate signal-to-noise regime, 
where signal is given by the size of the covariate effects
and noise is the variance level $\noise.$ 
Notably, under \Cref{condition:orthogonal_design}, $\noise$ relates directly to the variance of $\betaLS,$ and is influenced by both the residual variances and the dataset sizes; see \Cref{sec:identity_covariance_condition}.
In particular we interpret $\lambdaMin$ as a proxy for signal strength since it captures the magnitude of typical $\beta_d$'s along their direction of least variation.
\begin{corollary}\label{cor:gain_SNR}
$\gain \le  4\kappa^2\lambdaMin/\noise$
and $\gain \le 4\kappa^2(\lambdaMin/\noise)\inv,$
where $\kappa := \lambdaMax/\lambdaMin$ is the condition number of $\SBTrue.$
\end{corollary}
\Cref{cor:gain_SNR} formalizes the intuitive result that with enough noise, the little recoverable signal is insufficient to effectively share strength.
And furthermore, in the low-noise and high-signal regime $\betaID$ is very accurate on its own and there is little need for joint modeling.
However, when there is a large gap between the largest and smallest eigenvalues of $\SBTrue,$ leading $\kappa$ to be large, the gain could be larger.
$\kappa$ will be large, for example, when the covariate effects are very correlated across datasets.
 % \label{sec:asymptotic_gains}

%%%% Experiments
\section{Experiments}\label{sec:experiments}
\subsection{Simulated data}
We first conduct simulations, where we can directly control the relatedness among datasets and where we know the ground truth values of the parameters. We show that ECov is more accurate than EData when covariates outnumber datasets, whether effects are correlated across datasets or not.

In particular, we simulated covariates, parameters, and responses for $Q=10$ datasets across a range of covariate dimensions.
We generated covariate effects as $\beta_d\overset{i.i.d.}{\sim}\mathcal{N}(0, \SB)$. We chose $\SB$ so that effects were either 
correlated (\Cref{fig:simulation} Left) or independent (\Cref{fig:simulation} Right) across datasets; see \Cref{sec:experiments_supp} for details.
We compare performance of six estimates on these datasets.
These are estimates assuming EData/ECov using moment matching and maximum marginal likelihood to choose $\Sigma$/$\Gamma$ ($\betaEDMM$/$\betaMM$ and $\betaED$/$\betaEC,$ respectively),
as well as least squares ($\betaLS$), 
and ECov applied to each dataset independently ($\betaID$).

\begin{figure}
    \centering
    \includegraphics[width=0.96\textwidth]{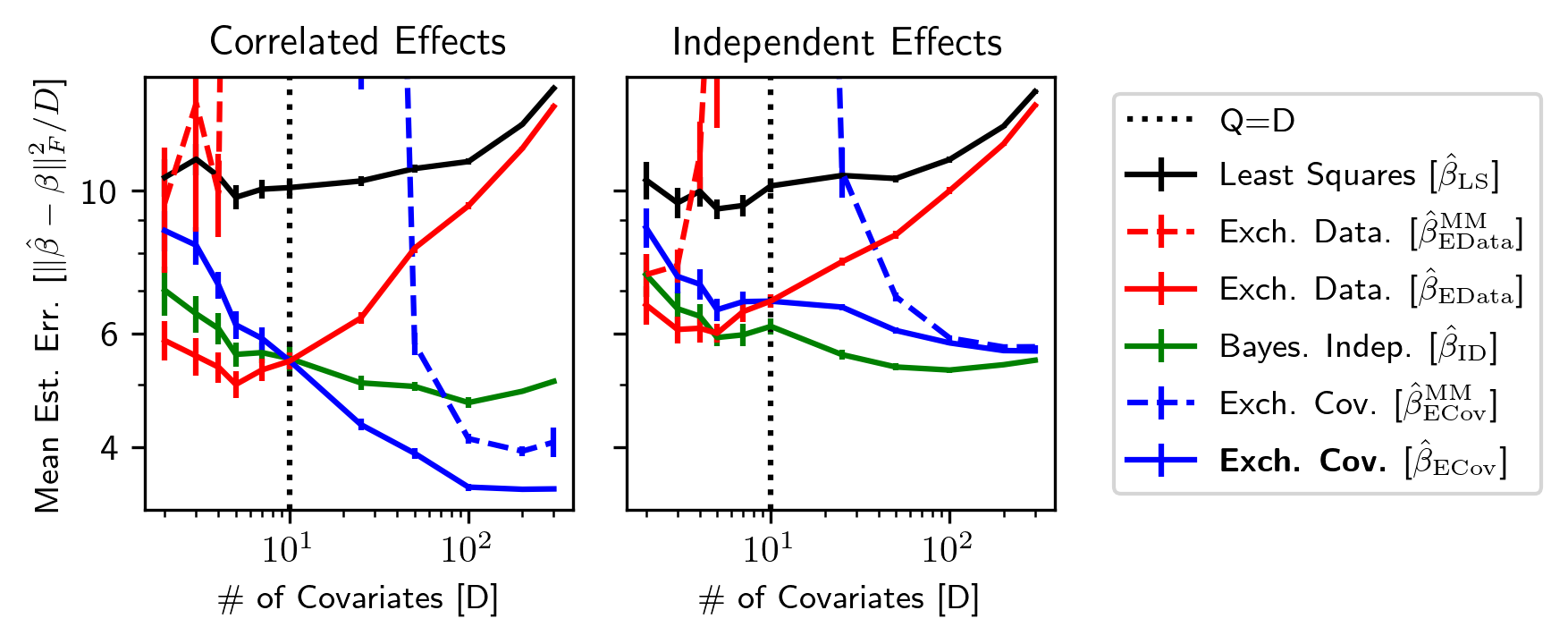}
    \caption{
Dimension dependence of parameter estimation error in simulation.
Covariate effects are either [Left] correlated or [Right] independent across the $Q=10$ datasets.
Each point is the mean $\pm 1 \texttt{SEM}$ across 20 replicates.}
\label{fig:simulation}
\end{figure}

\Cref{fig:simulation} reinforces our theoretical conclusions that (1) $\betaEC$ is more accurate when covariates outnumber datasets
and (2) $\betaED$ is more accurate when datasets outnumber covariates.
Our simulated $X$ matrices are somewhat relaxed from a strict orthogonal design (\Cref{sec:experiments_supp}),
so these experiments suggest that our conclusions may hold beyond \Cref{condition:orthogonal_design}.
Additionally, $\betaEC$ and $\betaED$ both outperform their moment based counterparts, $\betaMM$ and $\betaEDMM.$

Even for the simulations with independent effects,
\Cref{theorem:exch_reg_domination} suggests $\betaEC$ should still outperform $\betaLS$ and $\betaED$ in the higher dimensional regime,
and we see this behavior in the right panel of \Cref{fig:simulation}.
Additionally, in agreement with \Cref{theorem:asymptotic_gain}, $\betaEC$ does not improve over $\betaID$ in the presence of independent effects,
and the performances of these two estimators converge as $D$ grows.

\subsection{Real data}
We find that ECov beats EData, as well as least squares and independent estimation, across three real datasets. We describe the datasets (with additional details in \cref{sec:real_data_supp}) and then our results.

\textbf{Community level law enforcement in the United States.}
Policing rates vary dramatically across different communities, mediating disparate impacts of criminal law enforcement across racial and socioeconomic groups \citep{weisburd2019proactive,slocum2020enforcement}.
Understanding how demographic and socioeconomic attributes of communities relate to variation in rates of law enforcement is crucial to understanding these impacts. Linear models provide the desired interpretability.
We use a dataset \citep{redmond2002data} consisting of $D=117$ community characteristics and their rates of law enforcement (per capita) for different crimes.
We consider $Q=4$ data subsets corresponding to distinct (region, crime) pairs: (Midwest, Robbery), (South, Assault), (Northeast, Larceny), and (West, Auto-theft). This data setup illustrates a small $Q$ and accords with the independent residuals assumption in the likelihood shared by ECov and EData (\Cref{sec:exch_cov}).
Across $q$, $N^q$ represents between 400 and 600 communities.

\textbf{Blog post popularity.}
We regress reader engagement (responses) on $D=279$ characteristics of blog posts (covariates) \citep{buza2014feedback}.
We divided the corpus based on an included length attribute into $Q=3$ datasets,
corresponding to (1) long posts, (2) short posts, and (3) posts from an earlier corpus with missing length attribute.
We hypothesized that the relationships between the characteristics of posts and engagement would differ across these three datasets.
We randomly downsampled to $N^q=500$ posts in each dataset to mimic a low sample-size regime, in which sharing strength is crucial.

\begin{figure}[h]
    \centering
    \includegraphics[width=\textwidth]{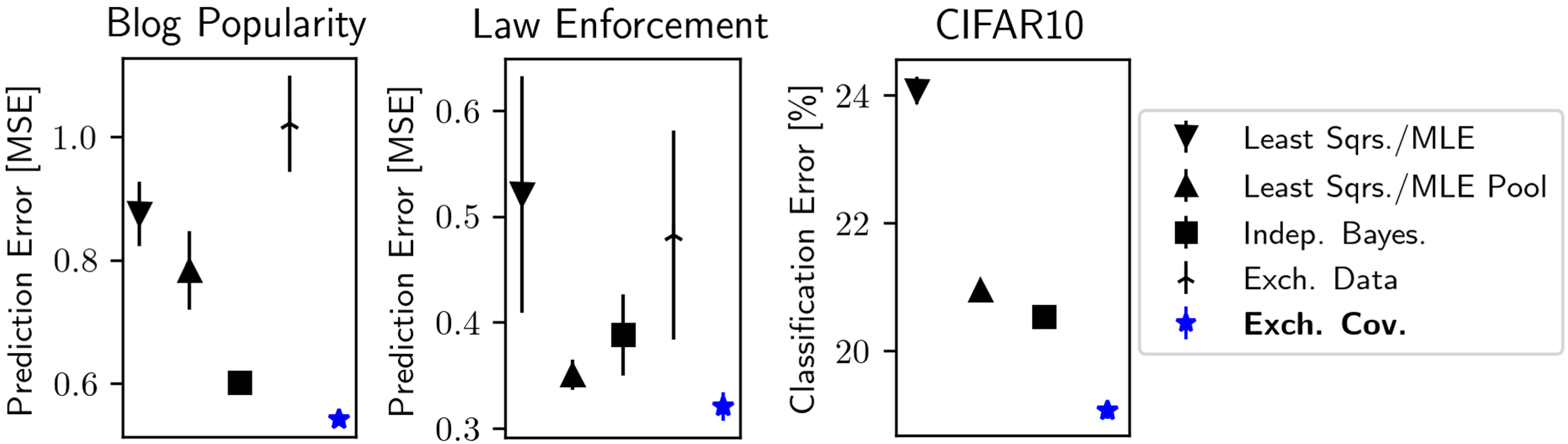}
    \caption{
        Prediction performance on held out data in three applications (mean $\pm 1\texttt{SEM}$ across 5-fold cross-validation splits).
    }\label{fig:applications}
\end{figure}

\textbf{Multiple binary classifications using pre-trained neural network embeddings on CIFAR10.}
Modern machine learning methods have proved very successful on large datasets. Translating this success to smaller datasets is one of the most actively pursued algorithmic challenges in machine learning.
It has spurred the development of frameworks from transfer learning \citep{weiss2016survey} to one-shot learning \citep{vinyals2016matching} to meta-learning \citep{finn2017model}.
One common and simple strategy starts with a learned representation (or ``embedding'') from an expressive neural network fit to a large dataset. Then one can use this embedding as a covariate vector for classification tasks with few labeled data points.

We take a $D=128$ dimensional embedding of the CIFAR10 image dataset \citep{krizhevsky2009learning,alibi2019detect}.
We create $Q=8$ different binary classification tasks using the classes in CIFAR10 (\cref{sec:real_data_supp}).
We downsampled to $N^q$ varying from 100 to 1000 to mimic a setting in which we hope to share strength from large datasets to improve performance on smaller datasets.

\textbf{Discussion of evaluation and results.}
In previous sections we have focused on parameter estimation. Here we instead evaluate with prediction error on held-out data since the true parameters are not observed.
Specifically we perform 5-fold cross-validation and report the mean squared errors and classification errors on test splits.
To reduce variance of out-of-sample error estimates on the applications in which we downsampled, we also evaluate on the additional held-out data.
Because the residual variances were unknown, we estimated these for each application and dataset as
$\hat \sigma^2_q := \|P_{X^q}^\perp Y^q\|^2/(N_q - D),$
where $P_{X^q}^\perp := I_{N_q} - X^q(X^{q\top}X^q)\inv X^{q^\top}$ (see e.g.\ \citep[Chapter 18.1]{gelman2006data}).
All methods ran quickly on a 36 CPU machine; computation of $\betaEC,$ including the EM algorithm, required 2.04 $\pm$ 0.64, 6.89 $\pm$ 3.19 and 37.14 $\pm$ 3.39 seconds (\texttt{mean} $\pm$ \texttt{st-dev} across splits) on the law enforcement, blog, and CIFAR10 tasks, respectively. 

Our results further reinforce the main aspects of our theory.
$\betaEC$ outperformed $\betaED$,
independent Bayes estimates ($\betaID$), and least squares ($\betaLS$) in all applications (at $> 95\%$ nominal confidence with a paired t-test).\footnote{
We did not develop an extension akin to \Cref{alg:E_step_logistic} for \textrm{EData}, and so do not report $\betaED$ for CIFAR10.
Additionally, we report a maximum likelihood estimate (MLE) instead of $\betaLS$ for CIFAR10.}
Additionally, $\betaEC$ outperformed the baseline of ignoring heterogeneity, pooling datasets together,
and using the same effect estimates for every dataset (``Least Sqrs./MLE Pool'').

\Cref{sec:experiments_supp} includes additional results and comparisons.
In particular, we provide the performance of the estimators on each component dataset for each application.
Additionally, we report the performances of
(1) stable and computationally efficient moment based alternatives to $\betaEC$ and $\betaED$
and (2) variants of $\betaEC$ and $\betaED$ that include a learned (rather than zero) prior mean.
\Cref{sec:licenses} reports the licenses of software we used.
 % \label{sec:experiments}

%%% Discussion
\section{Discussion}\label{sec:discussion}
The Bayesian community has long used hierarchical modeling with priors encoding exchangeability of effects across datasets (EData).
In the present work, we have made a case for instead using priors that encode exchangeability across \emph{covariates} (ECov) -- in particular, when the number of covariates exceeds the number of datasets.
We have presented a corresponding concrete model and inference method. We have
shown that ECov outperforms EData in theory and practice when the number of covariates exceeds the number of datasets.

Our approach is, of course, not a panacea. In some settings, a priori exchangeability among covariate effects will be inconsistent with prior beliefs. For example, imagine in the CIFAR10 application if meta-data covariates (such as geo-location and date) were available, in addition to embeddings. Then we might achieve better performance by treating meta-data covariates as distinct from embedding covariates.
Additionally, we focused on a Gaussian prior for convenience. In cases where practitioners have more specific prior beliefs about effects, alternative priors and likelihoods may be warranted, though they may be more computationally challenging.
Moreover, while relatively interpretable, linear models have their downsides. The linear assumption can be overly simplistic in many applications. It is common to misinterpret effects as causal rather than associative. Both the linear model and squared error loss lend themselves naturally to reporting means, but in many applications a median or other summary is more appropriate; so using a mean for convenience can be misleading.

%Nevertheless, we believe that, on net, our methodology is likely to produce positive societal impacts.
%Machine learning has recently trended toward black-box predictive models.
%Though linear models are already widely used in applied disciplines for their interpretability, we hope that our work will improve their accuracy and help further extend their reach.

Many exciting directions for further investigation remain.
For example, the covariance $\SB$ may provide an informative measure of task similarity; this similarity measure can be useful in, e.g., meta learning \citep{jerfel2019reconciling} and statistical genetics \citep{bulik_finucane2015atlas}.
It also remains to extend our methodology to other generalized linear models.
 % \label{sec:discussion}

\begin{ack}
The authors thank 
    Sameer K. Deshpande, 
    Ryan Giordano, 
    Alex Bloemendal,
    Lorenzo Masoero,
    and Diana Cai
for insightful discussions and comments on the manuscript.
This work was supported in part by ONR Award N00014-18-S-F006 and an NSF CAREER Award.
BLT is supported by NSF GRFP.
\end{ack}

\medskip

\small
\bibliography{references}

\normalsize

\newpage

\appendix
\section{Additional Related Work}\label{sec:additional_related_work}

\subsection{Brown and Zidek details}\label{sec:brown_and_zidek_details}
% Brown and Zidek
As discussed in \Cref{sec:intro}, the papers of \citet{brown1980adaptive} and \citet{haitovsky1987multivariate} carry the only references of which we are aware of the idea of exchangeability among covariate effects.
We here provide additional discussion on this related prior work.
To aid our comparison, we slightly modify their notation to match ours.

In their paper, ``Adaptive Multivariate Ridge Regression'', \citet{brown1980adaptive} consider multiple related regression regression problems with a shared design (i.e. $X:=X^1 =X^2=\cdots = X^Q$) and seek to extend the univariate ridge regression estimator of \citet{hoerl1970ridge} to the multivariate setting.
Specifically, the authors propose a class of estimators of the form
$$
\hat {\vec \beta}  = (I_Q \otimes X^\top X + K \otimes I_D)\inv (I_Q \otimes X^\top )\vec Y,
$$
where $\vec Y := [Y^{1\top}, Y^{2\top}, \cdots,Y^{Q\top}]^\top$, $\otimes$ denotes the Kronecker product,
and $K$ is a $Q \times Q$ ridge matrix which they suggest be chosen by some ``adaptive rule'' (i.e.\ that $K$ be a function of the observed data).
Notably, this functional form closely resembles our expression for $\E[ \vec \beta | \data, \SB]$ in \Cref{prop:conj_form}, if we take $K=\SB\inv.$

The authors do not explicitly discuss the interpretation of $K\inv$ as the covariance of a Gaussian prior,
nor any interpretation for this quantity as capturing any notion of a priori similarity of the regression problems.
However, they do point to Bayesian motivations at the outset of the paper.
In particular,
\citet{brown1980adaptive} narrow their consideration of possible methods for choosing $K$ to those which satisfy two criteria:
\begin{enumerate}
    \item{For any $K$, $\hat {\vec \beta}$ correspond to a Bayes estimate.}
    \item{In the case that $X^\top X=I_D$, $\hat {\vec \beta}$ correspond to the \citet{efron1972limiting}
    extension of the \citet{james1961estimation} estimator to vector observations.\footnote{
            See \Cref{sec:related_work_normal_means} for further discussion of connections to \citet{efron1972limiting}.}}
\end{enumerate}
They present four such estimators (derived from existing estimators of a multivariate normal means that dominate the sample mean) and demonstrate conditions under which each of these estimators dominates the least squares estimator for $\beta$.

As a further point of connection, the authors claim 
in the their abstract that their ``result is implicitly in the work of \citet{lindley1972bayes} although not actually developed there.''
However, the authors give little support for, or clarification of this claim.
In particular, their analysis is entirely frequentist and they provide no explanation for how their proposed estimators for $K$ might be interpreted as reasonable empirical Bayes estimates.

In their short follow-up paper, \citet{haitovsky1987multivariate} elaborates on this Bayesian motivation.
The primary focus of \citet{haitovsky1987multivariate} is a matrix normal prior \citep{dawid1981some} that captures structure in effects across both datasets and covariates.
Though this prior is not exchangeable across covariate effects in general, they note 
that the special case of where effects are uncorrelated across different covariates satisfies the notion of exchangeability for which we have advocated in this paper.
% Haitovsky
%\paragraph{Haitovsky.}
%In a short paper following up on \citet{brown1980adaptive}, \citet{haitovsky1987multivariate} elaborates on a hierarchical Bayesian motivation for the estimators of \citet{brown1980adaptive}.
%In particular the author considers Bayesian priors for $\beta$ with Kronecker structured prior covariance,
%$$
%\vec \beta \sim \mathcal{N}(0, \Sigma_1 \otimes \Sigma_2).
%$$
%
%The authors note that the special case in which $\Sigma_1=I$, this prior captures the exchangeability we advocate for here, i.e. a priori exchangeability of covariate effects.
%We provide an empirical comparison to these estimators in \Cref{sec:experiments}.
%Like \citet{brown1980adaptive}, \citet{haitovsky1987multivariate} restricts consideration to the case of shared designs (i.e. $X^1 = X^2=\cdots=X^Q$),
%and does not provide empirical evaluation.
%
%Both papers considers domination results over least squares estimators.
%
%\paragraph{Zelner SUR}
%Frequentist, focus on asymptotic efficiency (perhaps take from Stephens paper).

\subsection{Methods of inference for $\Gamma$ in existing work assuming exchangeability of effects across datasets.}\label{sec:inference_of_gamma_approaches}
We here describe several existing approaches for estimating the covariance matrix $\Gamma$ in the exchangeability of effects among datasets model.
These existing methods do not translate directly to the exchangeability of effects among covariates model proposed in this paper.
However, in principle, one could likely adapt any of them to our setting.
We have chosen to use the EM algorithm described in \Cref{sec:our_method} for its simplicity, efficiency, and stability.
We leave the investigation of alternative estimation approaches to future work.

In their initial paper, \citeauthor{lindley1972bayes} (\citeyear{lindley1972bayes}) \citep{lindley1972bayes} suggest that a fully Bayesian approach would be ideal.
They advocate for placing a subjectively specified, conjugate Wishart prior on $\Gamma,$ and remark that one should ideally consider the posterior of $\Gamma$ rather than relying on a point estimate.
However, in the face of analytic intractability, they propose returning MAP estimates for $\Gamma$  and $\beta$ and provide an iterative optimization scheme that they show is stationary at $\hat \Gamma, \hat \beta = \argmax \log p(\Gamma, \beta | \data).$

Advances in computational methods since 1972 have given rise to other ways of estimating $\Gamma$ in this model.
\citet{gelfand1990illustration} describe a Gibbs sampling algorithm for posterior inference.
\citet[Chapter 15 sections 4-5]{gelman2013bayesian} describe an EM algorithm which returns a maximum a posteriori estimate marginalizing over $\beta$, $\hat \Gamma = \argmax p(\Gamma | \data) = \int p(\Gamma,\beta | \data) d \beta$;
notably, though the updates in our EM algorithm for the case of exchangeability in effects across covariates differ from those in the case of exchangeability among datasets, one can see the two algorithms as closely related through their shared dependence on Gaussian conjugacy.
Finally, in the software package \texttt{lme4}, \citet{bates2015fitting} use the maximum marginal likelihood estimate, $\hat \Gamma = \argmax p(\data | \Gamma),$ which they compute using gradient based optimization.

%%% Additional related work on the normal means problem
\subsection{Related work on estimation of normal means}\label{sec:related_work_normal_means}
As we discuss in \Cref{sec:identity_covariance_condition}, under \Cref{condition:orthogonal_design} and when $\noise =1,$
we have that 
$$
\betaLS^q \overset{indep}{\sim}\mathcal{N}(\beta^q,I_D).
$$
As such, inference reduces to the ``normal means problem'', with a matrix valued parameter.
Specifically, we can equivalently write
$$
\betaLS = \beta + \epsilon,
$$
for a random $D\times Q$ matrix $\epsilon$ with i.i.d.\ standard normal entries.

This problem has been studied closely outside of the context of regression. 
Notably, \citet{efron1972empirical} approach the problem from an empirical Bayesian perspective 
and recommend an approach analogous to estimating $\SB$ by 
$$\SBhat^{\mathrm{Ef}} := (D-Q-1)\inv \betaLS^\top \betaLS - I_Q.$$
\citet{efron1972empirical} argue for this estimate because it is unbiased for a transformation of the parameter.
In particular, $\SBhat^{\mathrm{Ef}}$ satisfies $\E[(I_Q + \SBhat^{\mathrm{Ef}})\inv] = (I_Q + \SB)\inv$
when each $\beta_d \overset{i.i.d.}{\sim} \mathcal{N}(0, \SB).$
They show that, among all estimates of the form $\alpha \betaLS^\top \betaLS - I_Q$ with real valued $\alpha$, this factor $\alpha=(D-Q-1)\inv$ is optimal in terms of squared error risk.
Notably, this includes the moment estimate $\SBMM$ we describe in \Cref{sec:nonasymptotic_theory}, which corresponds to $\alpha=D\inv$.
However, this optimality result does not translate to the associated positive part estimators.
In fact, in experiments not shown, we have found that $\betaEC$ reliably outperforms an analogous positive part variant that estimates $\SB$ by $\SBhat^{\mathrm{Ef}}$.

\begin{remark}
\citet[Theorem 5]{efron1972empirical} prove that an analogous positive part estimator is superior to their original estimator in term of ``relative savings loss'' (RSL).
Our domination result in \Cref{theorem:exch_reg_domination} is strictly stronger and implies an improvement in RSL as well.
Furthermore our proof technique immediately applies to their estimator.
\end{remark}

Several other works have noted the dependence of the risk of estimators for the matrix variate normal means problem on the expectations of the eigenvalues of inverse non-central Wishart matrices \citep{efron1972empirical,zidek1978deriving,van1980multivariate}.
In all of these cases, the authors did not document attempts to interpret or approximate these difficult expectations.

More recently, \citet{tsukuma2008admissibility} explores a large class of estimators for the matrix variate normal means problems that shrink $\betaLS$ along the directions of its singular vectors in different ways.
For subclass of these estimators, \citet{tsukuma2008admissibility}[Corollary 3.1] proves a domination result for associated positive part estimators.
In the orthogonal design case, $\betaEC$ can be shown to be a member of this subclass of estimators, providing an alternative route to proving \Cref{theorem:pos_part_dominance}.

%%% More related Bayesian work...
\subsection{Additional related work on multiple related regressions}
Methods for simultaneously estimating the parameters of multiple related regression problems have a long history in statistics and machine learning,
with different assumptions and analysis goals leading to a diversity of inferential approaches. 
Perhaps the most famous is Zellner's landmark paper on seemingly unrelated regressions (SUR) \citep{zellner1962efficient}.
\citet{zellner1962efficient} addresses the situation where apparent independence of regression problems is confounded by covariance in the errors across $Q$ problems.
In the presence of such correlation in residuals, the parameter may be identified with greater asymptotic statistical efficiency 
by considering all $Q$ problems together \citep{zellner1962efficient,zellner1962further}.
While most work on SUR has taken a purely frequentist perspective in which $\beta$ is assumed fixed, some more recent works on SUR have considered Bayesian approaches to inference \citep{blattberg1991shrinkage,chib1995hierarchical,smith2000nonparametric,griffiths2003bayesian,ando2010hierarchical}.
However these do not address the scenario of interest here, in which we believe \textit{a priori} that there may be some covariance structure in 
the effects of covariates \textit{across} the regressions, or that some regression problems are more related than others.
The setting of the present paper further differs from SUR in that we do not consider correlation in residuals as a possible mechanism for sharing strength between datasets,
but instead explicitly assume independence in the noise.

\citet{breiman1997predicting} present a distinct, largely heuristic approach to multiple related regression problems where all $Q$ responses are observed for each dataset,
or equivalently each dataset has the same design.
The authors focus entirely on prediction
and obviate the need share information across regression problems when forming an initial estimate of $\beta$ by proposing to predict new responses in each regression with a linear combination of the predictions of linear models defined by the independently computed least squares estimate of each regression problem.
However this approach does not consider the problem of estimating parameters, which is a primary concern of the present work.
%The authors argue that the optimal choice of this linear combination can be derived a canonical correlation analysis of the population distributions of the covariates and responses.
%With only finite samples available, they advocate for performing a canonical analysis of the observed $X$ and $Y$, and applying additional shrinkage to the canonical correlations to compensate for biases induced by using a finite sample estimate.
%In particular they propose to choose these $Q$ shrinkage factors by cross-validation on the prediction mean-squared error.
%Cleverly, the authors show that this may be done analytically since the cross-validation objective is a quadratic form in the shrinkage factors.
%However, particularly in high dimensional settings it may be unreasonable to expect the individual least squares estimates to be sufficiently accurate for any linear combination to provide good performance.

\citet{reinsel1985mean}'s paper, ``Mean Squared Error Properties of Empirical Bayes Estimators in a Multivariate Random Effects General Linear Model'', considers a mixed effects model in which a linear model for regression coefficients is specified
$
\beta^q= B a_q + \lambda_q
$
where $a:=[a_1, a_2, \dots, a_Q]$ is a $K\times Q$ known design matrix associated with the regression problems,\footnote{
Notably, though \citet{reinsel1985mean} refers to $a$ as a design matrix, it has little relation of the design matrices $X^q$ to which we frequently refer in the present work.}
$B$ is a $D\times K$ matrix of unknown parameters and $[\lambda_1, \lambda_2, \dots, \lambda_Q]$ is a $D\times Q$ matrix of error terms.
These error terms are assumed exchangeable across datasets.
In contrast to the present work, \citet{reinsel1985mean} requires the relatedness between datasets to be known a priori through the known design matrix $a.$

\citet{laird1982random} consider a random effects model for longitudinal data in which different individuals correspond to different regression problems with distinct parameters.
In their construction, covariance structure in the noise is allowed across the observations for each individual, but not across individuals.
Additionally, as in \cite{lindley1972bayes}, the authors model the covariance in effects of different covariates a priori within each regression, but not covariance across regressions.

\citet{brown1998multivariate} propose to use sparse prior for $\beta$ which encourages a shared sparsity pattern.
Conditioned on a binary $D-$vector $\gamma \in \{0, 1\}^D$, $\beta$ is supposed to follow a multivariate normal prior as
$$
\vec \beta \overset{i.i.d.}{\sim} \mathcal{N}(0, \SB\otimes H_\gamma)
$$
where $H_\gamma$ is a $D \times D$ covariance matrix which expresses that for $d$ such that $\gamma_d=0$ we expect each $\beta_{d,q}$ to be close to zero.
Notably, this is equivalent to the assumption that $\beta$ follows a matrix-variate multivariate normal distributed as $\beta \sim \mathcal{MN}(0, H_\gamma, \SB)$ \citep{dawid1981some}.
Curiously, and without stated justification, the same $\SB$ is also taken to parameterize the covariance of the residual errors,
as well as of an additional bias term.
We suspect this restriction is made for the sake of computational tractability.
Indeed, \citep{stephens2013unified} makes similar modeling assumptions for tractability in the context of statistical genetics.
In contrast to the present work, the premise of \citet{brown1998multivariate} is sharing strength through 
similar sparsity patterns and covariance in the residuals, rather than learning and leveraging patterns of similarity in effects of covariates across datasets.

Other more recent papers have considered alternative approaches for multiple regression with sparse priors \citep{bhadra2013joint,lewin2015mt,deshpande2019simultaneous}.  These methods are of course inappropriate when we do not expect sparsity \textit{a priori}.

%%% Meta-learning
\paragraph{Meta-Learning}
The popular ``Model Agnostic Meta-Learning'' (MAML) approach \citep{finn2017model} can be understood as a hierarchical Bayesian method that treats tasks / datasets exchangeably \citep{grant2018recasting}.
As such, MAML and its variations do not allow tasks to be related to different extents (as our approach does).
A few recent works on meta-learning are exceptions;  
for example, \citet{jerfel2019reconciling} model tasks as grouped into clusters by using a Dirichlet process prior,
and \citet{cai2020weighted} consider a weighted variant of MAML that allows,
for a given task of interest, the contribution of data from other tasks to vary.
However these works differ from the present paper in their focus on prediction with flexible black-box models, 
whereas the primary concern of the present is parameter estimation in linear models.

\paragraph{Exchangeability of effects across covariates in the single dataset context.}
In the context of regression problems consisting of only a single dataset (i.e.\ corresponding to the special case of $Q=1$)
\citet{lindley1972bayes} suggest modeling the $D$ scalar covariate effects exchangeable.
In particular, they suggest modeling scalar covariate effects as i.i.d.\ from a univariate Gaussian prior when this exchangeability assumption is appropriate.
However, because this development is restricted to analyses of a single dataset, it does not relate to the problem of sharing strength across multiple datasets,
which is the subject of the present work.

\section{\Cref{sec:our_method} supplementary proofs and discussion}\label{sec:our_method_supp}

\subsection{Proof of \Cref{prop:conj_form}}\label{sec:proof_of_conj_form}
\begin{proof}
First note that the least squares estimates $\betaLS := [(X^{1\top}X^1)\inv X^{1\top}Y^1, \dots, (X^{Q\top}X^Q)\inv X^{Q\top}Y^Q]$ are a sufficient statistic of $\data$ for $\beta,$
and so $\beta |\data,\SB \sim \beta | \betaLS, \SB.$
As such, it is sufficient to consider the likelihood of $\betaLS.$
Let $\hat {\vec {\beta}}_{\mathrm{LS}} := [Y^{1\top} X^1 (X^{1\top}X^1)\inv , \dots, Y^{Q\top}X^Q(X^{Q\top}X^Q)\inv]$
be the $DQ$-vector defined by stacking the least squares estimates for each dataset.
Since for each $q,$ we have $\betaLS^q|\beta \overset{indep.}{\sim}\mathcal{N}(\beta^q, \sigma_q^2 (X^{q\top}X^q)\inv ),$
we can write $\hat {\vec {\beta}}_{\mathrm{LS}} | \beta \sim\mathcal{N}\left[\vec\beta, \diag\left(\sigma_1^2 (X^{1\top}X^1)\inv, \dots, \sigma_Q^2 (X^{Q\top}X^Q)\inv\right)\right].$
Next, that each $\beta_d\overset{i.i.d.}{\sim} \mathcal{N}(0, \SB)$ a priori implies that we may write
$\vec \beta \sim \mathcal{N}(0, \SB\otimes I_D)$ a priori, where $\otimes$ is the Kronecker product.
Then, by Gaussian conjugacy (see e.g.\ \citet[Chapter 2.3]{Bishop2006}), 
we have that $\vec \beta | \data \sim \mathcal{N}(\vec \mu, V),$ where 
$\vec \mu = V \left[ (\SB\otimes I_D)\inv 0 + \diag\left(\sigma_1^2 (X^{1\top}X^1)\inv, \dots,\sigma_Q^2 (X^{Q\top}X^Q)\inv\right)\inv \hat{\vec{\beta}}_{\mathrm{LS}} \right]
$
for 
$V\inv = (\SB\otimes I_D)\inv + \diag\left(\sigma_1^2 (X^{1\top}X^1)\inv, \dots,\sigma_Q^2 (X^{Q\top}X^Q)\inv\right)\inv.$
Due to the block structure of the matrices above, these simplify to 
$\vec \mu = V \left[\frac{Y^{1\top}X^1}{\sigma^2_1}, \dots, \frac{Y^{Q\top}X^Q}{\sigma^2_Q}\right]$ and
$V\inv= \SB\inv \otimes I_D +  \diag(\frac{X^{1\top}X^1}{\sigma_1^2}, \dots,\frac{X^{Q\top}X^Q}{\sigma_Q^2}),$
as desired.
\end{proof}

\subsection{Efficient computation with the conjugate gradient algorithm}\label{sec:conj_gradients}
As mentioned in \Cref{sec:conjugate_gaussian_estimate}, $\vec \mu=\E[\vec\beta|\data, \SB]$ in \Cref{prop:conj_form} may be computed efficiently using the conjugate gradient algorithm (CG) for solving linear systems.
We here describe several properties of CG that make it surprisingly well-suited to this application.

We first note that \Cref{prop:conj_form} allows us to frame computation of $\vec \mu$ as the solution to the linear system
$$
A\vec \mu=b
$$
for $b=\left[Y^{1\top}X^1/\sigma^2_1, \dots, Y^{Q\top}X^Q/\sigma^2_Q\right]^\top$ and 
$A=\SB\inv \otimes I_D +  \diag\left(\sigma_1^{-2}X^{1\top}X^1, \dots,\sigma_Q^{-2} X^{Q\top}X^Q\right).$
A naive approach to computing $\vec \mu$ could then be to explicitly compute $A\inv$ and report the matrix vector product, $A\inv b.$
However, as mentioned in \Cref{sec:conjugate_gaussian_estimate}, since $A$ is a $DQ\times DQ$ matrix, 
explicitly computing its inverse would require roughly $O(D^3Q^3)$ time.
This operation becomes very cumbersome when $D$ and $Q$ are too large;
for instance if $D$ and $Q$ are in the hundreds the, $DQ$ is is the tens of thousands.

CG provides an exact solution to linear systems in at most $DQ$ iterations, with each iteration requiring only a small constant number of matrix vector multiplications by $A$.
This characteristic does not provide a complexity improvement for solving general linear systems because
for dense, unstructured $DQ\times DQ$ matrices, matrix vector multiplies require $O(D^2Q^2)$ time, and CG still demands $O(D^3Q^3)$ time overall.
However this property provides a substantial benefit in our setting.
In particular, the special form of $A$ allows computation of matrix vector multiplications in $O(D^2Q)$ rather than $O(D^2Q^2)$ time,
and storage of this matrix with $O(D^2Q)$ rather than $O(D^2Q^2)$ memory.
Specifically, if $v=[v_1, v_2, \dots, v_Q]$ is a $D\times Q$ matrix with $D$-vector columns $v_q,$
for the $DQ$-vector $\vec v=[v_1^\top, v_2^\top, \dots, v_Q^\top]^\top$ we can compute 
$A\vec v$ as $\texttt{vec}\left( v \SB\inv\right) + [\sigma_1^{-2}X^{1\top}X^1 v_1, \dots, \sigma_Q^{-2}X^{Q\top}X^Q v_Q]^\top,$
where $\texttt{vec}(\cdot)$ represents the operation of reshaping an $D\times Q$ matrix into a $DQ$-vector by stacking its columns.
When $D>Q,$ this operation is dominated by the $Q\ \ O(D^2)$ matrix-vector multiplications to compute the second term.
As such, CG provides an order $Q$ improvement in both time and memory.

Next, CG may be viewed as an iterative optimization method.
At each step it provides an iterate which is the closest to the $\vec \mu$ on a Krylov subspace of expanding dimension.
As such, the algorithm may be terminated after fewer than $DQ$ steps to provide an approximation of the solution.
Moreover, the algorithm may be provided with an initial estimate, and improves upon that estimate in each successive iteration.
In our case we may readily compute a good initialization.
For example, we can initialize with the posterior mean of the parameter for each dataset when conditioning on that dataset alone, i.e.\ $\vec \mu^{(0)} := \left[ \E[\beta^1|Y^1]^\top, \dots , \E[\beta^Q|Y^Q]^\top \right]^\top.$

Finally, the convergence properties of the conjugate gradient algorithm are well understood.
Notably the $i$th iterate of conjugate gradient $\vec \mu^{(i)}$ when initialized at $\vec \mu^{(0)}$ satisfies
$$
\|\vec \mu^{(i+1)}- \vec \mu\|_A \le 2\left(\frac{
            \kappa - 1
        }{
            \kappa + 1
    }\right)^i \|\vec \mu^{(0)} - \vec \mu\|_A,
$$
where $\kappa=\sqrt{\frac{\lambda_{\text{max}}(A)}{\lambda_\text{min}(A)}}$ is the square root of the condition number of $A$, and $\| \cdot \|_A$ is the $A-$quadratic norm \citep[Chapter 5.1]{nocedal2006numerical}, \citep{luenberger1973introduction}.
Since $A$ will often be reasonably well conditioned (note, for example, that $\lambda_\text{min}(A) \ge \lambda_{\text{min}}(\SB)$),
convergence can be rapid.
Notably, in an unpublished application the authors encountered (not described in this work) 
involving $D\approx 20,000$ covariates and $Q \approx 50$ datasets, 
the approximately million dimensional estimate $\vec \mu$ was computed 
in roughly 10 minutes on a 16 core machine.

\subsection{Expectation maximization algorithm further details}\label{sec:em_details}
In \Cref{sec:empirical_bayes_and_EM,sec:logistic_regression_empirical_bayes_main} we introduced EM algorithms for estimating $\SB$ for both linear and logistic regression models.
In this subsection we provide a derivation of the updates in \Cref{alg:EM_general} and discuss computational details of our fast implementation.

\paragraph{Derivations of EM updates for linear regression.}
Our notation inherits directly from \citep[Chapter 1.5]{mclachlan2007algorithm}, 
to which we refer the reader for context.
In our application of the EM algorithm, we take the collection of all covariate effects $\beta$ as the `missing data.'
For the expectation (E) step, we therefore require
\begin{align}\label{eqn:linear_EM_update}
\begin{split}
Q(\SB,\SB^{(i)}) :&=  \E[\log p(\beta| \SB) |\data,  \SB^{(i)}] \\
&= c +\frac{D}{2}\log |\SB\inv|-\frac{1}{2} \sum_{d=1}^D \E[\beta_d^\top \SB\inv \beta_d | \data, \SB^{(i)}] \\
&= c+\frac{D}{2}\log |\SB\inv|- \frac{1}{2} \sum_{d=1}^D \trace\left(\SB\inv\E[\beta_d\beta_d^\top | \data, \SB^{(i)}] \right)\\
&= c+\frac{D}{2}\log |\SB\inv|- \frac{1}{2} \sum_{d=1}^D \trace\left(\SB\inv(\mu_d\mu_d^\top + V_d) \right),
\end{split}
\end{align}
where $c$ is a constant that does not depend on $\SB,$
$\mu=[\mu_1 \dots, \mu_D]^\top := \E[\beta  | \data, \SB^{(i)}]$ and for each $d$ \ 
$V_d := (I_Q \otimes e_d)^\top \Var[\vec\beta | \data, \SB^{(i)}] (I_Q \otimes e_d).$
From the last line of \Cref{eqn:linear_EM_update} we may see that $\mu$ and $\{V_d\}_{d=1}^D,$ comprise the required posterior expectations.

The solution to the maximization step may then be found by considering a first order condition for maximizing over $\SB\inv$ rather than $\SB.$
Observe that
$\frac{\partial}{\partial \SB\inv} Q(\SB, \SB^{(i)})  
= \frac{D}{2}\SB  - \frac{1}{2}\sum_{d=1}^D (\mu_d\mu_d^\top + V_d).$
Setting this to zero we obtain $\SB^{(i+1)} = D\inv \sum (\mu_d\mu_d^\top +  V_d).$
This is the desired update for the M-step provided in \Cref{alg:E_step_conj}.

%The EM algorithm for LR is actually an EM algorithm, and it converges (jeff and wu).

\paragraph{Logistic regression EM updates.}
The updates for the approximate EM algorithm described in \Cref{sec:our_method} are derived from a Gaussian approximation to the posterior 
under which the expectation of log prior is taken.
In particular we approximate the first line of \Cref{eqn:linear_EM_update} as
\begin{align}\label{eqn:logistic_EM_update}
\begin{split}
Q(\SB,\SB^{(i)}) :&=  \E[\log p(\beta| \SB) |\data,  \SB^{(i)}] \\
                  &=  \int p(\beta | \data,   \SB^{(i)}) \log p(\beta| \SB) d \beta  \\
                  &\approx  \int q^{(i)}(\beta) \log p(\beta| \SB) d \beta 
\end{split}
\end{align}
where $q^{(i)}$ denotes the Laplace approximation to $p(\beta| \data, \SB^{(i)}).$
Specifically, as we summarized in \Cref{alg:E_step_logistic}, we approximate the posterior mean by the maximum a posteriori estimate, 
$\vec \mu^* := \argmax_{\vec\beta} \log p(\vec\beta | \data, \SB^{(i)}),$ 
and the posterior variance by
$V := -[\nabla_\beta^2 \log p(\vec\beta | \data, \SB^{(i)})\big|_{\vec \beta = \vec \mu^*}]\inv.$ 
We the let $q^{(i)}$ be the Gaussian density with these moments.
This renders the integral in the last line of \Cref{eqn:logistic_EM_update} tractable,
and updates are derived in the same way as in the linear case.

Naively, the approximate EM algorithm for logistic regression could be much more demanding than its counterpart in the linear case.
In particular, at each iteration we need to solve a convex optimization problem, rather than linear system. 
However, in practice the algorithm is only little more demanding because, 
by using the maximum a posteriori estimate from the previous iteration to initialize the optimization, we can solve the optimization problem very easily.
In particular, after the first few EM iterations, only one or two additional Newton steps from this initialization are required.

To simplify our implementation, we used automatic differentiation in \texttt{Tensorflow} to compute gradients and Hessians when computing the maximum a posteriori values and Laplace approximations.

\paragraph{Computational efficiency.}
We have employed several tricks to provide a fast implementation of our EM algorithms.
The M-Steps for both linear and logistic regression involve a series of expensive matrix operations.
To accelerate this, we used \texttt{Tensorflow}\citep{abadi2016tensorflow} to optimize these steps by way of a computational graph representation generated using the \texttt{@tf.function} decorator in python.
Additionally, we initialize EM with a moment based estimate (see \Cref{sec:real_data_moments}).

\section{Frequentist properties of exchangeability among covariate effects -- supplementary proofs and discussion}\label{sec:nonasymptotic_supp}

\subsection{Discussion of \Cref{condition:orthogonal_design}}\label{sec:identity_covariance_condition}
The restriction on the design matrices in \Cref{condition:orthogonal_design} places strong limits the immediate scope of our theoretical results.
However, as with many statistical assumptions such as Gaussianity of residuals,
this condition lends considerable tractability to the problem that enables us to build insights that we can see hold in more relaxed settings in experiments (see \Cref{sec:experiments}).

Under \Cref{condition:orthogonal_design} estimation of the parameter $\beta$ may be reduced to a special matrix valued case of the normal means problem with each
$\hat \beta_{\mathrm{LS},d}^q \sim \mathcal{N}(\beta_d^q, \noise).$
Accordingly, we may recognize $\noise$ as a reflection of both the residual variances $\sigma^2_q$ and sample sizes $N_q.$
In particular, if within each dataset $q$ the covariates have sample second moment $N_q\inv \sum_{n=1}^{N_q} X^q_n X^{q\top}_n=I_D,$ 
and the residual variances and sample sizes are equal (i.e.\ $\sigma_1^2=\sigma_2^2=\dots=\sigma^2_Q$ and $N^1=N^2=\dots=N^Q$),
then $\noise = \sigma_1^2/N^1.$
Additionally, because $\betaLS$ is a sufficient statistic of $\data$ for $\beta,$
it suffices to consider $\betaLS$ alone, without needing to consider other aspects of $\data.$
For these reasons, conditions of this sort are commonly assumed by other authors in related settings (e.g.\ \citet[Chapters 1.4 and 6.2]{van2015lecture} and \citet{fan2001variable,golan2002comparison}).

That the trends predicted by our theoretical results persist beyond the limits of \Cref{condition:orthogonal_design} 
should not be surprising.
The likelihood, our estimators and their risks are all continuous in the $X^q,$ 
and so domination results may be seen to extends via continuity to settings with well-conditioned designs.
On the other hand, problems with design matrices that are more poorly conditioned are more challenging for both theory and estimation in practice (see e.g. \citet{brown1980adaptive}[Example 4.2]).

\subsection{A proposition on analytic forms of the risks of moment estimators}\label{sec:moment_ests_supp}
The following proposition characterizes analytic expressions for the moment based estimators.
These expressions provide a starting point for the theory in \Cref{sec:nonasymptotic_theory}

\begin{prop}\label{prop:moments}
    Assume each $Y^q_n | X_n^q, \beta^q \sim \mathcal{N}(X_n^{q\top} \beta^q, \sigma_q^2)$ and define
$\SBMM := D\inv \betaLS^\top \betaLS 
- D\inv\diag(\sigma_1^2 \|X^{1\dagger}\|_F^2, \dots, \sigma_Q^2 \|X^{Q\dagger}\|_F^2).$
Then
\begin{enumerate}
    \item{if each $\beta_d \overset{i.i.d.}{\sim}\mathcal{N}(0, \SB),$ $\E[\SBMM]=\SB.$}
\end{enumerate}
Furthermore, under \Cref{condition:orthogonal_design}
\begin{enumerate}
  \setcounter{enumi}{1}
\item{when $D\ge Q,\ \ \betaMM = \betaLS- \noise D\betaLS^{\dagger\top}$ and}
\item{when $D\le Q,\ \ \betaEDMM = \betaLS- \noise Q\betaLS^{\dagger\top},$}
\end{enumerate}
where $\dagger$ denotes the Moore-Penrose pseudoinverse of a matrix.
\end{prop}
\begin{proof}
We begin with statement (1), that under \Cref{condition:orthogonal_design} and correct prior specification, 
$\E[\SBMM ]=\SB.$
Recall that $\SBMM := D\inv \betaLS^\top \betaLS 
- D\inv\diag(\sigma_1^2 \|X^{1\dagger}\|_F^2, \dots, \sigma_Q^2 \|X^{Q\dagger}\|_F^2).$
For any fixed $\beta,$ we have 
$\E[\SBMM |\beta ]= D\inv \E [\betaLS^\top \betaLS | \beta ]
- D\inv\diag(\sigma_1^2 \|X^{1\dagger}\|_F^2, \dots, \sigma_Q^2 \|X^{Q\dagger}\|_F^2),$
and so seek to characterize $\E[\betaLS^\top \betaLS | \beta ].$
Note that we may write $\betaLS \overset{d}{=} \beta  + \epsilon$ for a random $D\times Q$ matrix $\epsilon$ with each column $q$ distributed as 
$\epsilon^q \overset{indep.}{\sim} \mathcal{N}\left[0, \noise_q (X^{q\top} X^q)\inv \right].$
As such, for each $q$ we have $\E[\betaLS^{q\top}\betaLS^q | \beta]= \beta^{q\top}\beta^q + \E[\epsilon^{q\top}\epsilon^q].$
Next observe that $\E[\epsilon^{q\top}\epsilon^q]=\trace[\noise_q (X^{q\top} X^q)\inv]=\noise_q \|X^{q\dagger}\|_F^2,$
where $\dagger$ denotes the pseudo-inverse of a matrix and $\|\cdot\|_F$ is the Frobenius norm.
Additionally, for $q\ne q^\prime,$ we have $\E[\betaLS^{q\top}\betaLS^{q^\prime} | \beta]=\beta^{q\top}\beta^{q^\prime}.$
Putting these together into matrix form, we see $\E[ \betaLS^\top \betaLS | \beta ] = \beta^\top \beta + \diag(\sigma_1^2 \|X^{1\dagger}\|_F^2, \dots, \sigma_Q^2 \|X^{Q\dagger}\|_F^2),$
and so $\E[\SBMM |\beta ]= D\inv \beta^\top \beta.$
Under the additional assumption that for each $d,\, \beta_d\overset{i.i.d.}{\sim} \mathcal{N}(0, \SB),$
we have that $\E[D\inv \beta^\top \beta] = \SB,$ and (1) obtains from the law of iterated expectation.

We next prove statement (2), that $\betaMM := \E[\beta|\data, \SBMM]=\betaLS - \noise D \betaLS^{\dagger\top}.$
Consider the singular value decomposition (SVD), $\betaLS = V\diag(\lambda\msqrt)U^\top.$ 
Under \Cref{condition:orthogonal_design} substituting this expression into $\SBMM$ provides $\SBMM =D\inv U\diag(\lambda)U^\top - \noise I_Q.$
Therefore, \Cref{lemma:conj_form_under_condition} provides that we may write
\(
\betaMM &:= \E[\beta|\data, \SBMM] \\
        &= \betaLS - \betaLS\left[ \precision \SBMM + I_Q\right]\inv \\
        &= \betaLS - V\diag(\lambda\msqrt)U^\top \left[ \precision (D\inv U\diag(\lambda)U^\top - \noise I_Q) + I_Q\right]\inv U^\top \\
        &= \betaLS - V\diag\left[\lambda\msqrt \odot  ( \precision D\inv \lambda)\inv \right]U^\top \\
        &= \betaLS - \noise D V\diag(\lambda^{-\frac{1}{2}})U^\top \\
        &= \betaLS - \noise D \betaLS^{\dagger\top},
\)
where $\odot$ is the Hadamard (i.e.\ elementwise) product, as desired.

We lastly prove (3), that the analogous moment based estimator constructed under the assumption of a priori exchangeability among datasets is 
$\betaEDMM = \betaLS - \noise Q \betaLS^{\dagger\top}.$
We begin by making explicit the assumed model and estimate.
Specifically we assume each $\beta^q \overset{i.i.d.}{\sim} \mathcal{N}(0, \Gamma)$ a priori,
where $\Gamma$ is a $D\times D$ covariance matrix.

In this case, we obtain an unbiased moment based estimate of $\Gamma$ as 
$\hat \Gamma^{\mathrm{MM}} := Q\inv \betaLS\betaLS^\top - Q\inv \sum_{q=1}^Q \noise_q(X^{q\top}X^q)\inv.$
Following an argument exactly parallel to the one in the proof of (1),
we find that under the prior $\beta^q\overset{i.i.d.}{\sim} \mathcal{N}(0, \Gamma),$ we have 
$\E[\hat \Gamma^{\mathrm{MM}}]=\Gamma.$
Furthermore, following an argument exactly parallel to the one in the proof of (2),
we find that under \Cref{condition:orthogonal_design} the corresponding empirical Bayes estimate 
$\betaEDMM := \E[\beta | \hat \Gamma^{\mathrm{MM}}] = \betaLS - \noise Q\betaLS^{\dagger\top}.$
We omit full details to spare repetition.
\end{proof}

\begin{lemma}\label{lemma:conj_form_under_condition}
Under \Cref{condition:orthogonal_design}
$\E[\beta |\data, \SB] = \betaLS - \betaLS\left[\precision\SB + I_Q \right]\inv.$
\end{lemma}
\begin{proof}
    By \Cref{prop:conj_form}, we have 
\(
    \E[\vec \beta |\data, \SB] = 
V \left[\frac{Y^{1\top}X^1}{\sigma^2_1}, \dots, \frac{Y^{Q\top}X^Q}{\sigma^2_Q}\right]
\ \text{  where  }\ 
V\inv= \SB\inv \otimes I_D +  \diag(\frac{X^{1\top}X^1}{\sigma_1^2}, \dots,\frac{X^{Q\top}X^Q}{\sigma_Q^2}).
\)
Under \Cref{condition:orthogonal_design}, we can simplify this as 
\(
\E[\vec \beta |\data, \SB] &= 
\left[\SB\inv \otimes I_D +  \diag(\frac{X^{1\top}X^1}{\sigma_1^2}, \dots,\frac{X^{Q\top}X^Q}{\sigma_Q^2})\right]\inv
 \left[\frac{Y^{1\top}X^1}{\sigma^2_1}, \dots, \frac{Y^{Q\top}X^Q}{\sigma^2_Q}\right]\\
 &= \left[\SB\inv \otimes I_D +  \precision I_{DQ}\right]\inv \precision \left[\betaLS^1, \dots, \betaLS^Q\right]\\
 &= \left[\noise \SB\inv \otimes I_D +  I_{DQ}\right]\inv \left[\betaLS^1, \dots, \betaLS^Q\right].
\)
As a result, for each $d,\, \E[\beta_d | \data, \SB] =\left[ \noise \SB\inv + I_Q\right] \inv \beta_{\mathrm{LS},d}$ and so,
in matrix form, we may write
\(
\E[\beta |\data, \SB] &=  \betaLS \left[ \noise \SB\inv + I_Q\right]\inv  \\
&= \betaLS - \betaLS\left[I_Q + \precision\SB\right]\inv.
\)

\end{proof}

%%% Proof of form of risks lemma
\subsection{Proof of \Cref{lemma:form_of_risks}}
\begin{proof}
We prove the lemma in two parts;
first for the case that $D>Q+1,$
and then for the case that $Q\le D\le Q+1.$

Our proof for the case that $D>Q+1$ relies on an expression for the squared error risk for estimators of the form
$\hat \beta = \betaLS -\noise c\betaLS^{\dagger \top}$
for real $c.$
In particular, \Cref{lemma:SURE_with_pseudo_inv} provides that when $D>Q+1$ and under \Cref{condition:orthogonal_design},
$$\E[\|\beta - (\betaLS - c\betaLS^{\dagger\top})\|_F^2 \mid \beta ] = DQ + \sigma^4 c(c +2 +2Q -2D)\E[\|\betaLS^\dagger\|_F^2 \mid \beta ].$$
Notably, since under \Cref{condition:orthogonal_design}, by \Cref{prop:moments} we have that $\betaMM=\betaLS - \noise D\betaLS^{\dagger\top}$
we obtain
$\E[\|\beta - \betaMM\|_F^2 \mid \beta ] = \noise DQ - \sigma^4 D(D - 2Q -2)\E[\|\betaLS^\dagger\|_F^2 \mid \beta ],$
as desired.

We next consider $Q \le D \le Q+1.$
In this case, both $\risk(\beta, \betaMM)$ and $\noise DQ - \sigma^4 D(D - 2Q -2)\E[\|\betaLS^\dagger\|_F^2 \mid \beta ]$ are positive infinity.
In particular, observe that $\|\betaLS^\dagger\|_F^2 = \trace[(\betaLS^\top \betaLS)\inv]$ is the trace of the inverse of a non-central Wishart matrix, which is known to have infinite expectation for $Q\le D \le Q+1$ (see e.g.\ \citet{hillier2019properties}).
Likewise, \Cref{lemma:infinite_risk} reveals that $\risk(\beta, \betaMM)=\infty$ as well.

The second assertion of \Cref{lemma:form_of_risks},
that when $D\le Q$ and under \Cref{condition:orthogonal_design}
$\E[\|\beta - \betaEDMM\|_F^2 \mid \beta] = \noise DQ - \sigma^4 Q(Q - 2D -2)\E[\|\betaLS^\dagger\|_F^2 \mid \beta],$
obtains similarly.
Specifically, under these conditions an identical argument to that provided in \Cref{lemma:SURE_with_pseudo_inv} provides that 
$$\E[\|\beta - (\betaLS - \noise c \betaLS^{\dagger\top})\|_F^2\mid \beta] = DQ + \sigma^4 c(c +2 +2D -2Q)\E[\|\betaLS^\dagger\|_F^2\mid \beta]$$
when $D<Q-1.$
The desired expression is then obtained by taking $c=Q$ to reflect $\betaEDMM = \betaLS - \noise Q\betaLS^{\dagger\top},$
again as specified by \Cref{prop:moments}.
\end{proof}

\begin{lemma}\label{lemma:SURE_with_pseudo_inv}
Let $D>Q+1$ and let $\hat \beta = \betaLS - \noise c \betaLS^{\dagger\top}.$
Then under \Cref{condition:orthogonal_design}
$\E[\|\beta - \hat \beta\|_F^2\mid \beta] = \noise DQ + \sigma^4 c(c + 2 + 2Q - 2D)\E[ \|\betaLS^\dagger \|_F^2\mid \beta].$
\end{lemma}
\begin{proof}
The results follows by considering Stein's unbiased risk estimate (SURE) \citep[Chapter 4, Corollary 7.2]{lehmann2006theory} (restated as \Cref{lemma:SURE_lehmann}) and making several algebraic simplifications.
In order to apply the lemma, we note that under \Cref{condition:orthogonal_design} 
$\hat {\vec \beta}_{\mathrm{LS}} 
\sim \mathcal{N}(\vec \beta, \noise I_{DQ})$ and $\hat { \vec\beta} = \hat {\vec \beta}_{\mathrm{LS}} - g(\hat {\vec \beta}_{\mathrm{LS}})$ for 
$g(\hat {\vec \beta}_{\mathrm{LS}}) = -\noise c\cdot \texttt{vec}(\betaLS^{\dagger\top})$,
where $\texttt{vec}(\cdot)$ represents the operation of reshaping an $D\times Q$ matrix into a $DQ$-vector by stacking its columns.

We first simplify the sum of partial derivatives in \Cref{eqn:SURE} of \Cref{lemma:SURE_lehmann}. 
Observe that 
$$\sum_{n=1}^{DQ}
\frac{\partial g_{n}(\hat {\vec \beta}_{\mathrm{LS}})}{\partial \hat {\vec \beta}_{\mathrm{LS},n}} 
=-\noise c \sum_{d=1}^D \sum_{q=1}^Q 
\frac{\partial\hat \beta_{\mathrm{LS},d}^{\dagger, q} }{\partial \hat \beta_{\mathrm{LS},d}^{q}},
$$
where $\hat \beta_{\mathrm{LS},d}^{\dagger, q}$ denotes the entry in the $q$th row and $d$th column of $\betaLS^\dagger.$

Next, letting $e_q$ be the $q$th basis vector in $\R^Q,$
for each $q$ and $d$ we may write 
\(
\frac{\partial\hat \beta_{\mathrm{LS},d}^{\dagger, q} }{\partial \hat \beta_{\mathrm{LS},d}^{q}} 
&=\frac{\partial}{\partial \hat \beta_{\mathrm{LS},d}^{q}} \hat \beta_{\mathrm{LS},d}(\betaLS^\top \betaLS)\inv e_q \\
&=e_q^\top (\betaLS^\top \betaLS)\inv e_q +
\hat \beta_{\mathrm{LS},d}\frac{\partial}{\partial \hat \beta_{\mathrm{LS},d}^{q}} (\betaLS^\top \betaLS)\inv e_q \\
&= e_q^\top (\betaLS^\top \betaLS)\inv e_q  -
\hat \beta_{\mathrm{LS},d}^\top(\betaLS^\top\betaLS)\inv
\left[\frac{\partial}{\partial \hat \beta_{\mathrm{LS},d}^{q}} (\betaLS^\top\betaLS)\right]
 (\betaLS^\top\betaLS)\inv e_q \\
 &=\|\betaLS^{\dagger,q}\|^2-
 \hat \beta_{\mathrm{LS},d}^{\dagger\top}\left[e_q \hat\beta_{\mathrm{LS},d}^\top  + 
 \hat\beta_{\mathrm{LS},d} e_q^\top \right] (\betaLS^\top\betaLS)\inv e_q \\
 &=\|\betaLS^{\dagger,q}\|^2- 
\left[ \hat \beta_{\mathrm{LS},d}^{\dagger\top} e_q \hat \beta_{\mathrm{LS},d}^\top(\betaLS^\top\betaLS)\inv e_q +
\hat \beta_{\mathrm{LS},d}^{\dagger\top} \hat \beta_{\mathrm{LS},d} e_q^\top(\betaLS^\top\betaLS)\inv e_q \right]\\
 &=\|\betaLS^{\dagger,q}\|^2- 
(\hat \beta_{\mathrm{LS},d}^{\dagger,q})^2 -
\hat \beta_{\mathrm{LS},d}^{\dagger\top} \hat \beta_{\mathrm{LS},d}\|\betaLS^{\dagger,q}\|^2,
\)
where in the fourth and last lines we have used that $e_q^\top (\betaLS^\top \betaLS)\inv e_q = \|\betaLS^{\dagger,q}\|^2,$ as can be seen by observing that $(\betaLS^\top\betaLS)\inv = \betaLS^\dagger\betaLS^{\dagger\top} .$

Adding these terms together we find
\(
\sum_{d=1}^D\sum_{q=1}^Q \frac{\partial\hat \beta_{\mathrm{LS},d}^{\dagger, q} }{\partial \hat \beta_{\mathrm{LS},d}^{q}} 
&= \sum_{d=1}^D\sum_{q=1}^Q \Big\{
\|\betaLS^{\dagger,q}\|^2- 
(\hat \beta_{\mathrm{LS},d}^{\dagger,q})^2 -
\hat \beta_{\mathrm{LS},d}^{\dagger\top} \hat \beta_{\mathrm{LS},d}\|\betaLS^{\dagger,q}\|^2
\Big\}\\
 &= D\|\betaLS^{\dagger}\|_F^2- 
\|\betaLS^{\dagger}\|_F^2- 
 \|\betaLS^{\dagger}\|_F^2 \sum_{d=1}^D\hat \beta_{\mathrm{LS},d}^{\dagger\top} \hat \beta_{\mathrm{LS},d} \\
 &= D\|\betaLS^{\dagger}\|_F^2- 
\|\betaLS^{\dagger}\|_F^2- 
\|\betaLS^{\dagger}\|_F^2 \trace(\betaLS^\dagger \betaLS)\\
&= (D-Q-1)\|\betaLS^{\dagger}\|_F^2.
 \)

We next note that the regularity condition required by \Cref{lemma:SURE_lehmann} is satisfied, as demonstrated in \Cref{lemma:SURE_regularity_condition},
and so we may write
\(
\E[\|\beta - \hat \beta\|_F^2\mid\beta] &= \noise DQ + \E[\|g(\hat {\vec \beta}_{\mathrm{LS}}) \|^2\mid\beta] 
- 2 \noise \sum_{d=1}^D\sum_{q=1}^Q \E[\frac{\partial\hat \beta_{\mathrm{LS},d}^{\dagger, q} }{\partial \hat \beta_{\mathrm{LS},d}^{q}}\mid \beta]  \\
&= \noise DQ +  \sigma^4 c^2 \E[\|\betaLS^\dagger\|^2\mid\beta] - 2 \sigma^4 c (D-Q-1)\E[\|\betaLS^{\dagger}\|_F^2\mid \beta] \\
&= \noise DQ + \sigma^4 c (c +2 +2Q -2D)\E[\|\betaLS^\dagger\|^2\mid\beta].
\)
as desired.
\end{proof}

\begin{lemma}[Stein's Unbiased Risk Estimate -- Lehmann and Casella Corollary 7.2]\label{lemma:SURE_lehmann}
Let $X\sim \mathcal{N}(\theta, \noise I_N),$ 
and let the estimator $\hat \theta$ be of the form
$\hat \theta = X - g(X)$
where $g(X) = [g_1(X), g_2(X), \dots, g_N(X)]$ is differentiable.
If $\E[|\frac{\partial }{\partial X_n}g_n(X)|] < \infty$ for each $n=1,\dots, N,$
then 
\[\label{eqn:SURE}
\risk(\theta, \hat \theta) = \noise N + \E[\|g(X)\|^2] - 2 \noise \sum_{n=1}^N \frac{\partial }{\partial X_n}g_n(X).
\]
\end{lemma}

\begin{lemma}\label{lemma:SURE_regularity_condition}
Let $D>Q+1.$ 
Then under \Cref{condition:orthogonal_design}
$\E\left[\left|\frac{\partial\hat \beta_{\mathrm{LS},d}^{\dagger, q} }{\partial \hat \beta_{\mathrm{LS},d}^{q}}\right| \mid \beta \right] \le \infty
$
for each $d$ and $q.$
\end{lemma}
\begin{proof}
From our derivation of 
$\frac{\partial\hat \beta_{\mathrm{LS},d}^{\dagger, q} }{\partial \hat \beta_{\mathrm{LS},d}^{q}}$ 
in \Cref{lemma:SURE_with_pseudo_inv} 
we have that 
\(
\frac{\partial\hat \beta_{\mathrm{LS},d}^{\dagger, q} }{\partial \hat \beta_{\mathrm{LS},d}^{q}}
 &=\|\betaLS^{\dagger,q}\|^2- 
(\hat \beta_{\mathrm{LS},d}^{\dagger,q})^2 -
\hat \beta_{\mathrm{LS},d}^{\dagger\top} \hat \beta_{\mathrm{LS},d}\|\betaLS^{\dagger,q}\|^2\\
 &=\|\betaLS^{\dagger,q}\|^2- 
(\hat \beta_{\mathrm{LS},d}^{\dagger,q})^2 -
\|\betaLS^{\dagger,q}\|^2\trace[(\betaLS^\top\betaLS)\inv\beta_{\mathrm{LS},d}\beta^\top_{\mathrm{LS},d}].
\)
As such we have that 
\(
\left|\frac{\partial\hat \beta_{\mathrm{LS},d}^{\dagger, q} }{\partial \hat \beta_{\mathrm{LS},d}^{q}}\right| 
&\le 
\|\betaLS^{\dagger,q}\|^2+ 
|(\beta_{\mathrm{LS},d}^{\dagger,q})^2|+
\|\betaLS^{\dagger,q}\|^2|\trace[(\betaLS^\top\betaLS)\inv\beta_{\mathrm{LS},d}\beta^\top_{\mathrm{LS},d}]| \\
&\le \|\betaLS^{\dagger,q}\|^2+ 
\left|\sum_{d^\prime=1}^D (\beta_{\mathrm{LS},d^\prime}^{\dagger,q})^2\right|+
\|\betaLS^{\dagger,q}\|^2\left|\trace[(\betaLS^\top\betaLS)\inv
\sum_{d^\prime=1}^D \beta_{\mathrm{LS},d^\prime}\beta_{\mathrm{LS},d^\prime}^\top]\right| \\
&= \|\betaLS^{\dagger,q}\|^2+ 
    \|\betaLS^{\dagger,q}\|^2 +
    \|\betaLS^{\dagger,q}\|^2 \trace[(\betaLS^\top\betaLS)\inv \betaLS^\top\betaLS] \\
    &\le (2+Q)\|\betaLS^{\dagger,q}\|^2\\
    &\le (2+Q)\|\betaLS^{\dagger}\|_F^2\\
    &= (2+Q)\trace[(\betaLS^\top \betaLS)\inv].
\)
We next recognize that under \Cref{condition:orthogonal_design}, $(\betaLS^\top \betaLS)\inv$ is the inverse of a non-central Wishart matrix with non-centrality parameter $\beta.$
Therefore, from \citet[Theorem 1]{hillier2019properties}, we have that for $D>Q+1$,  $\E\left[\trace\left((\betaLS^\top \betaLS)\inv\right)\mid \beta\right]< \infty.$
Accordingly, we may conclude that 
$\E\left[\left|\frac{\partial\hat \beta_{\mathrm{LS},d}^{\dagger, q} }{\partial \hat \beta_{\mathrm{LS},d}^{q}}\right| \mid \beta \right] \le \infty$
as desired.
\end{proof}

\begin{lemma}\label{lemma:infinite_risk}
Assume $Q\le D\le Q+1.$
For any $\beta, \, \risk(\beta, \betaMM) = \infty.$
\end{lemma}
\begin{proof}
First observe that we may lower bound $L(\beta, \betaMM)$ as 
\(
L(\beta, \betaMM)
&=\|\betaMM - \beta\|_F^2 \\
    &=\|\noise D\betaLS^{\dagger\top} +\beta - \betaLS\|_F^2 \\
    &=\sigma^4 D^2 \|\betaLS^{\dagger}\|_F^2 +\|\beta - \betaLS\|_F^2 
    - 2 \noise D \trace\left[-\betaLS^{\dagger}(\beta - \betaLS)\right] \\
    &\ge \sigma^4 D^2 \|\betaLS^{\dagger}\|_F^2 +\|\beta - \betaLS\|_F^2 
    - 2 \noise D \|\betaLS^{\dagger}\|_F\|\beta - \betaLS\|_F \\
    &= (\noise D \|\betaLS^{\dagger}\|_F - \|\beta - \betaLS\|_F)^2 
\)
where the inequality follows from Cauchy-Schwarz.
We next consider any constant $c<\noise D$ and write
\(
\risk(\beta, \betaMM) 
&= \E[L(\beta, \betaMM)|\beta]\\
    &= \pr(c\|\betaLS^\dagger\|_F\ge \|\betaLS -\beta \|_F) \E[L(\beta, \betaEDMM)\mid \beta,c \|\betaLS^\dagger\|_F\ge \|\betaLS -\beta \|_F] \\
    &+\pr(c\|\betaLS^\dagger\|_F< \|\betaLS -\beta \|_F) \E[L(\beta, \betaEDMM)\mid \beta,c\|\betaLS^\dagger\|_F < \|\betaLS -\beta \|_F]\\
    &\ge \pr(c\|\betaLS^\dagger\|_F\ge \|\betaLS -\beta \|_F) \E[L(\beta, \betaEDMM)\mid \beta,c\|\betaLS^\dagger\|_F\ge \|\betaLS -\beta \|_F]\\
    &\ge \pr(c\|\betaLS^\dagger\|_F\ge \|\betaLS -\beta \|_F) \E[
        (\noise D \|\betaLS^{\dagger}\|_F - \|\beta - \betaLS\|_F)^2
        \mid \beta,c\|\betaLS^\dagger\|_F\ge \|\betaLS -\beta \|_F]\\
    &\ge \pr(c\|\betaLS^\dagger\|_F\ge \|\betaLS -\beta \|_F) (\noise D-c)^2 \E[\|\betaLS^{\dagger}\|_F^2
        \mid\beta,c\|\betaLS^\dagger\|_F\ge \|\betaLS -\beta \|_F]\\
    &\ge  (\noise D-c)^2 \pr(c\|\betaLS^\dagger\|_F\ge \|\betaLS -\beta \|_F) \E[\trace[(\betaLS^\top \betaLS)\inv
        \mid\beta] =\infty
\)
where the last line comes from recognizing $(\betaLS^\top\betaLS)\inv$ as the inverse of a non-central Wishart matrix, 
the trace of which has infinite expectation for $Q\le D \le Q+1.$
\end{proof}

%%% Additional detail related to the moment est. domination results
\subsection{Proof of \Cref{theorem:exch_reg_domination} and additional details}

\begin{proof}
The first domination result of \Cref{theorem:exch_reg_domination} follows closely from \Cref{lemma:form_of_risks}.
Under \Cref{condition:orthogonal_design}, $\betaLS \overset{d}{=} \beta + \sigma\epsilon$ for a random matrix $\epsilon$ with i.i.d.\ standard normal entries, and so we can see $\risk(\beta, \betaLS)= \sum_{d=1}^D \sum_{q=1}^Q \E[(\sigma\epsilon^q_d)^2]=DQ\noise.$ 
Next, $D>2Q+2$ implies that $D-2-2Q>0$ so that
$D(D-2-2Q)\noise \|\betaLS^\dagger\|_F^2$ is almost surely positive,
and therefore positive in expectation.
We therefore obtain the result from \Cref{lemma:form_of_risks}.

We next consider the second domination result.
The performance of $\betaEDMM$ may be seen to degrade in stages as we transition from a few covariates and many datasets regime to a many covariates and few datasets regime.
When $D<Q/2 -1,$ we can see that $\betaEDMM$ has good performance.
In fact, by an argument analogous to our proof of the first part of \Cref{theorem:exch_reg_domination} above, we can see that $\betaEDMM$ dominates $\betaLS$;
Specifically, from \Cref{lemma:form_of_risks} we can recognize $\risk(\beta, \betaLS) - \risk(\beta, \betaED)$ as the expectation of an almost surely positive quantity.

When $D= Q/2 - 1$ we have $Q(Q-2-2D)=0,$ and so regardless of $\beta,$ the estimators $\betaEDMM$ and $\betaLS$ have equal risk, and neither dominates.

Relative performance degrades further in the intermediate regime of $Q/2 -1  < D  < Q-1.$
In this regime, 
$\risk(\beta, \betaLS) - \risk(\beta, \betaEDMM) = \sigma^4 Q(Q-2-2D)\E[\|\betaLS^\dagger\|_F^2 \mid \beta ]$
may be written as the expectation of an almost surely negative quantity, and so $\betaEDMM$ is dominated by $\betaLS.$

The situation is even worse when $Q-1\le D \le Q;$
appealing again the they symmetry between $\betaEDMM$ and $\betaMM,$ we can see that by \Cref{lemma:infinite_risk} $\risk(\beta, \betaEDMM)=\infty.$

Finally, when $D>Q$ the expression $\betaEDMM =\betaLS - \left[ \precision \GammaMM  - I_D \right]\inv \betaLS$
involves the inverse of a low rank matrix since under \Cref{condition:orthogonal_design}, $\GammaMM= Q\inv \betaLS\betaLS^\top - \noise I_D.$
Accordingly we take as our convention $\| \betaEDMM\|=\infty,$ analogously to defining $\frac{1}{0}=\infty;$ 
as a result $\betaEDMM$ has infinite risk in this second regime as well, and we see that 
this estimator is dominated by $\betaLS$ whenever $D < Q/2 -1.$

\end{proof}

With the strong parallels established by \Cref{prop:moments,lemma:form_of_risks} under \Cref{condition:orthogonal_design},
we can see that this is not a result of $\betaEDMM$ being singularly bad.
Indeed, if we consider the large datasets regime with $Q>D,$ we can obtain analogous results to demonstrate the superiority of an exchangeability among datasets approach.

\subsection{Proof of \Cref{lemma:form_of_mle_est}}
\begin{proof}
We first show that under \Cref{condition:orthogonal_design}, 
$\SBhat= U\diag\left[ (D\inv\lambda - \noise \ones{Q})_+\right]U^\top $
is the maximum marginal likelihood estimate of $\SB$ in \Cref{eqn:empirical_bayes_estimate}.
Our approach is to first derive a lower bound on the negative log likelihood,
and then show that this bound is met with equality by the proposed expression.

For convenience, we consider a scaling of the negative log likelihood, 
$$
-2D\inv \ln p(\betaLS|\SB)  = 
\ln |\SB + \noise I_Q | + D\inv \trace\left[ (\SB + \noise I_Q)\inv \betaLS^\top \betaLS\right],
$$
and are interested in deriving a lower bound on
$$
\min_{\SB\succeq 0} \ln |\SB + \noise I_Q | + D\inv \trace\left[ (\SB + \noise I_Q)\inv \betaLS^\top \betaLS\right],
$$
where the notation $\SB\succeq 0$ reflects that the minimum is taken over the space of positive semidefinite matrices.

The problem simplifies if we parameterize the minimization with the eigendecomposition $\SB = V^\top \diag(\nu)V,$
where $V$ is a $Q\times Q$ matrix satisfying $V^\top V = I_Q$ and $\nu$ is a $Q$-vector of non-negative reals.
In particular, if we define $\mathcal{L}(V, \nu) := -2D\inv \ln p(\betaLS|\SB=V^\top \diag(\nu)V)$ then,
leaving the constraints on $V$ and $\nu$ implicit, we have
\(
\min_{V, \nu} \mathcal{L}(V, \nu)
&=\min_{V, \nu} \ln |V^\top \diag(\nu)V+ \noise I_Q | + D\inv\trace\left[ (V^\top\diag(\nu)V+\noise  I_Q)\inv \betaLS^\top \betaLS\right] \\
&=\min_{V, \nu} \ln |V^\top \diag(\nu)V+ \noise I_Q | + D\inv \trace\left[ (\diag(\nu)+ \noise I_Q)\inv V \betaLS^\top \betaLS V^\top \right] \\
&=\min_{V, \nu} \sum_{q=1}^Q \ln (\nu_q + \noise) + D\inv \sum_{q=1}^Q  \frac{1}{\nu_q+\noise} V_q^\top \betaLS^\top \betaLS V_q \\
&=\min_{V} \sum_{q=1}^Q \min_{\nu_q\ge 0} \ln (\nu_q + \noise) +  \frac{D\inv V_q^\top \betaLS^\top \betaLS V_q}{\nu_q+\noise }.
\)
Next, \Cref{lemma:optimum} provides that we may solve the inner optimization problems over $\nu$ in the line above analytically
to get $\nu^* := \argmin_{\nu} \mathcal{L}(V, \nu)$ with entries $\nu^*_q=\max(\noise, D\inv V_q^\top \betaLS^\top \betaLS V_q) - \noise.$
Substituting these values in, we obtain
\(
\min_{V, \nu} \mathcal{L}(V, \nu)
&=\min_{V} \sum_{q=1}^Q 
    \ln \left[\max(\noise, D\inv V_q^\top \betaLS^\top \betaLS V_q)\right] +
    \frac{D\inv V_q^\top \betaLS^\top \betaLS V_q}{\max(\noise,D\inv V_q^\top \betaLS^\top \betaLS V_q)}\\
&=\min_{V} \sum_{q=1}^Q 
    \ln \left[\max(\noise, D\inv V_q^{\top} \betaLS^\top \betaLS V_q)\right] +
    \precision \min(\noise, D\inv V_q^{\top} \betaLS^\top \betaLS V_q).
\)
We can now further simplify the problem by considering the eigendecomposition of $\betaLS^\top \betaLS = U\diag(\lambda)U^\top,$
and recognizing that because $VU$ satisfies $(VU)^\top VU=I_Q$ we may write
\(
\min_{V, \nu} \mathcal{L}(V, \nu)
&=\min_{V} \sum_{q=1}^Q 
    \ln \left[\max(\noise, D\inv V_q^{\top} \betaLS^\top \betaLS V_q)\right] +
    \precision \min(\noise, D\inv V_q^{\top} \betaLS^\top \betaLS V_q)\\
&=\min_{V} \sum_{q=1}^Q 
    \ln \left[\max\left(\noise, V_q^\top \diag(D\inv \lambda)V_q \right)\right] +
    \precision \min\left[\noise, V_q^\top \diag(D\inv \lambda)V_q \right].
\)
Finally, we obtain a lower bound by recognizing $\{ V_q^\top \diag(D\inv \lambda)V_q \}_{q=1}^Q$
as the diagonals of $D\inv V \diag(\lambda)V^\top$ and applying \Cref{lemma:schur_and_majorization}
to obtain that 
$$
-2D\inv \ln p(\betaLS|\SB)  \ge 
\sum_{q=1}^Q \ln \left[\max(\noise, D\inv \lambda_q)\right] + \precision \min(\noise, D\inv \lambda_q)
$$
for every $\SB\succeq 0.$

We next show that this bound is met with equality by $\SBhat = U\diag\left[ (D\inv\lambda - \noise \ones{Q})_+\right]U^\top,$
the form given in the statement of \Cref{lemma:form_of_mle_est}.
Recognize first that $\SBhat +\noise I_Q = U \diag\left[\max(\noise\ones{Q}, D\inv \lambda)\right]U^\top.$
Substituting this expression in, we find
\(
    -2D\inv\ln p(\betaLS|\SBhat) &= 
    \ln |\SBhat + \noise I_Q | + D\inv \trace\left[ (\SBhat + \noise I_Q)\inv \betaLS^\top \betaLS\right] \\
    &= \ln \left|\diag\left[\max(\noise \ones{Q}, D\inv \lambda)\right] \right| + 
D\inv \trace\left[ \diag\left[\max(\noise\ones{Q}, D\inv \lambda)\right]\inv U^\top \betaLS^\top \betaLS U \right] \\
&= \sum_{q=1}^Q \ln \left[\max(\noise ,D\inv \lambda_q) \right]+ 
    D\inv \lambda_q/\max(\noise , D\inv \lambda_q) \\
    &= \sum_{q=1}^Q \ln \left[\max(\noise, D\inv \lambda_q) \right]+ \precision \min(\noise, D\inv \lambda_q),
\)
which meets our lower bound.
This establishes that the maximum marginal likelihood estimate is $\SBhat = U\left[(D\inv \lambda -\noise \ones{Q})_+ \right]U^\top,$ as desired.

It now remains to show that, under \cref{condition:orthogonal_design}, $\betaEC = V\diag\left[\lambda\msqrt\odot (\ones{Q} - \noise D\lambda\inv)_+ \right]U^\top.$
By \Cref{lemma:conj_form_under_condition}, we have that $\betaEC = \betaLS - \betaLS\left[I_Q + \precision\SBhat\right]\inv.$
Substituting in the analytic expression for $\SBhat,$ recalling the SVD $\betaLS = V\diag(\lambda\msqrt)U^\top,$ 
and rearranging, we obtain
\(
\betaEC &= V\diag(\lambda\msqrt)U^\top -V\diag(\lambda\msqrt)U^\top\left\{
I_Q + \precision U\left[(D\inv \lambda -\noise \ones{Q})_+ \right]U^\top \right\}\inv \\
&= V\diag\left\{\lambda\msqrt-\lambda\msqrt\left[
\ones{Q} + \precision (D\inv \lambda -\noise \ones{Q})_+ \right]\inv\right\}U^\top  \\
&= V\diag\left\{\lambda\msqrt \odot \left[\ones{Q} -
\left(\ones{Q}  + (\precision D\inv \lambda - \ones{Q})_+ \right)\inv\right]\right\}U^\top  \\
&= V\diag\left[\lambda\msqrt \odot \left(\ones{Q} -\noise D \lambda\inv\right)_+ \right]U^\top, 
\)
as desired.
\end{proof}

\begin{lemma}\label{lemma:optimum}
    For any $c>0,$
\(
    \nu^* :&= \argmin_{\nu\ge 0} \ln(\nu + \noise ) + \frac{c}{\nu +\noise }\\
    &= \max(\noise,c) - \noise
\)
\end{lemma}
\begin{proof}
Define $g(x) :=  \ln(x + \noise) + c/(x + \noise)$ and 
$f(x):=g(\noise x)=\ln(x + 1) + \frac{\precision c}{x + 1} +\ln \noise$ to lighten notation.
Now $\nu^* = \argmax_{x\ge 0} g(x) = \noise \argmax_{x\ge 0} f(x).$
Denote by $f^\prime$ and $f^{\prime\prime}$ the first two derivatives of $f.$
Notably, $f^\prime(x) = (x+1)\inv\left[1 - \precision c/(x+1)\right]$
and $f^{\prime\prime}(x) = (x+1)^{-2}\left[2\precision c/(x+1) - 1\right].$
The result may be seen by separately considering the cases of $\precision c<1$ and $\precision c\ge 1.$

If $\precision c<1$, then $f^\prime$ is positive on $\R_+,$ and so $\argmin_{x\in\R_+} f(x)=0.$
On the other hand, if $\precision c\ge 1,$ then $f$ has a local minimum at $x=\precision c-1$ (note that $f^\prime(\precision c-1)=0$, and $f^{\prime\prime}(\precision c-1)>0)$).
Since this is the only local minimum on $\R_+,$ and with the positive second derivative at the this minimum, we can conclude that in this case $\argmin_{x\in\R_+}f(x)=\precision c -1.$
In either case, we can write $\argmin_{x\in\R_+}f(x)=\max(1,\precision c) - 1.$
Therefore, as desired, we see that $\argmin_{x\in \R_+}g(x) = \max(\noise, c) - \noise.$
\end{proof}

\begin{lemma}\label{lemma:schur_and_majorization}
Let $ A $ be a $Q\times Q$ Hermitian matrix with eigenvalues $\lambda_1, \lambda_2, \dots, \lambda_Q.$
Then
$$\sum_{q=1}^Q 
        \ln \left[\max(\noise, A_{q,q})\right] +
        \precision \min(\noise, A_{q,q})
    \ge \sum_{q=1}^Q \ln \left[\max(\noise, \lambda_q)\right] + \precision \min(\noise, \lambda_q).
$$
\end{lemma}
\begin{proof}
First note that $f(x) = \ln \max(\noise, x) + \min(\noise, x)$ is concave on $\R_+,$ and so the vector valued function, 
$g(x_1, x_2, \dots, x_N) = \sum_{n=1}^N f(x_n)$ is Schur concave.
By the Schur-Horn theorem (\Cref{theorem:schur_horm}) the diagonals of $A$ are majorized by its eigenvalues, when each are sorted in descending order.
As such $g\left(\diag(A)\right) \ge  g\left(\lambda\right)$, as desired.
\end{proof}

%% Proof of positive part domination Theorem
\subsection{Proof of \Cref{theorem:pos_part_dominance}}
Our approach to showing dominance of $\betaEC$ over $\betaMM$ parallels the classical approach of \citet{baranchik1964multiple},
to showing that the positive part James-Stein estimator dominates the original James-Stein estimator.
In this case, however, our parameter and estimates are matrix-valued, rather than vector-valued.
Additionally, we contend with the added complication that the directions along which we apply shrinkage are random.

\begin{proof}
To begin, consider again the SVD of the matrix of least squares estimates,
$\betaLS = V\diag(\lambda\msqrt)U^\top.$
Recall from \Cref{prop:moments} that $\betaMM = \betaLS - \noise D \betaLS^{\dagger\top}$
under \Cref{condition:orthogonal_design}.
Because the pseudo-inverse of $\betaLS$ may be written as 
$\betaLS^\dagger = U\diag(\lambda\nsqrt)V^\top,$
we rewrite $\betaMM = V\diag(\lambda\msqrt - \noise D\lambda\nsqrt)U^\top.$
Comparing this estimate to the expression for $\betaEC$ in \Cref{lemma:form_of_mle_est},
$\betaEC = V \diag\left[\lambda\msqrt\odot (1-\noise D \lambda\inv )_+\right]U^\top,$
we see that the two estimates differ only when $\betaMM$ ``flips the direction'' of one or more of the singular values of $\betaLS.$
Our strategy to proving the theorem is to show that analogously to the ``over-shrinking'' of the James-Stein estimator 
relative to the positive part James-Stein estimator, this ``over-shrinking'' of singular values increases the loss of $\betaMM$ in expectation.

For convenience, we define $\rho:=\lambda \msqrt\odot(1-\noise D\lambda\inv)$ and 
$\rho_+ := \lambda\msqrt \odot (1-\noise D \lambda\inv)_+$ so that 
$\betaMM = V \diag(\rho)U^\top$ and 
$\betaEC = V \diag(\rho_+)U^\top.$

To show the desired uniform risk improvement we must show that for any $\beta$,
\[\label{eqn:mle_uniform_loss}
\E\left[ \loss(\beta, \betaMM) - \loss(\beta, \betaEC)\right]  > 0,
\]
where $\loss(\beta, \hat \beta) = \|\hat \beta -\beta\|_F^2$ is squared error loss. 
We can rewrite this difference in loss as
\(
\loss(\beta, \betaMM) - \loss(\beta, \betaEC) &= \| \betaMM - \beta\|_F^2 -  \| \betaEC - \beta\|_F^2 \\
&= \| \diag(\rho) - V^\top \beta U\|_F^2 -  \| \diag(\rho_+) - V^\top \beta U\|_F^2 \\
&= \sum_{q=1}^Q (\rho_q - V_q^\top \beta U_q)^2 -  (\rho_{+q}- V_q^\top \beta U_q)^2\\
&= \sum_{q=1}^Q \rho_q^2 - \rho_{+q}^2 -2(V_q^\top \beta U_q)(\rho_q - \rho_{+q}),
\)
where we here (and in the proof of this theorem only) write $V_q$ and $U_q$ to denotes columns of $V$ and $U,$
rather than rows.
Since $\rho_{q}^2 \overset{a.s.}{\ge} \rho_{+q}^2$, it suffices to show that for any $\beta$ and each $q$, 
$$
\E \left[(V_q^\top \beta U_q)(\rho_q - \rho_{+q}) \right]< 0.
$$

To show this, we again find an even narrower but easier to prove condition will imply the one above;
since $\rho_q$ and $\rho_{+q}$ differ only when $\lambda_q < \noise D$, it is enough to show that for each $0 < c < \noise D$
\[\label{eqn:conditionally_negative}
\E\left[ (V_q^\top \beta U_q)\rho_q | \lambda_q=c \right]< 0.
\]
If we establish \Cref{eqn:conditionally_negative}, then \Cref{eqn:mle_uniform_loss} obtains from the law of iterated expectation.
Next, observe that since $\rho_q$ fixed and negative when $\lambda_q= c < \noise D,$
\Cref{eqn:mle_uniform_loss} is equivalent to
$$
\E\left[ V_q^\top \beta U_q| \lambda_q=c \right]> 0.
$$

Letting $U_{-q}$ and $V_{-q}$ denote the remaining columns of $U$ and $V$, respectively, we may write
\(
\E\left[ V_q^\top \beta U_q| \lambda_q=c \right] = 
\E\left[ \E\left[V_q^\top \beta U_q| \lambda_q=c, U_{-q}, V_{-q}\right] \right] 
\)
and, again through the law of iterated expectation, see that it will be sufficient to show for every $U_{-q}$ and $V_{-q}$
that $\E\left[V_q^\top \beta U_q| \lambda_q=c, U_{-q}, V_{-q}\right]> 0$.

With all but one column of each of $U$ and $V$ fixed, $U_q$ and $V_q$ are determined up to signs, as unit vectors in the one dimensional subspaces orthogonal to $[\{U^{q^\prime} \}_{q^\prime \ne q}]$ and $[\{V^{d}\}_{d\ne q}]$.
As such, we need only to show
\[\label{eqn:probably_positive}
    \pr\left[V_q^\top \beta U_q>0 | U_{-q}, V_{-q}, \lambda_q=c\right] 
> \pr\left[V_q^\top \beta U_q <0 | U_{-q}, V_{-q}, \lambda_q=c\right],
\]
since 
\(\E&\left[V_q^\top \beta U_q| \lambda_q, U_{-q}, V_{-q}\right]\\
    &=|V_q^\top \beta U_q| \left\{\pr\left[V_q^\top \beta U_q>0| \lambda_q, U_{-q}, V_{-q}\right] -
\pr\left[V_q^\top \beta U_q<0| \lambda_q, U_{-q}, V_{-q}\right]\right\},
\)
where, in an abuse of notation, we have moved 
$|V_q^\top \beta U_q| $ outside the expectation since it is deterministic once we have observed $V_{-q}$ and $U_{-q}$.

That \Cref{eqn:probably_positive} holds may be seen from considering the conditional probability densities for $U_q$ and $V_q$,
and noting that the density is larger for $V_q$ and $U_q$ such that $V_q^\top \beta U_q$ is positive.
In particular, we have that
\(
\ln p(\betaLS|\beta, U_{-q}, V_{-q}, \lambda) &= -\frac{1}{2}\| \beta - \hat \beta\|_F^2 + h\\
&= -\frac{1}{2}\|V^\top \beta U -  \diag(\lambda\msqrt) \|_F^2 + h\\
&= -\frac{1}{2} (\lambda_q\msqrt - V_q^\top \beta U_q )^2 + h^\prime
\)
where $h$ and $h^\prime$ are constants that do not depend on the signs of $U_q$ and $V_q.$
Since $\lambda_q\msqrt$ is positive with probability one, the conditional probability that $V_q^\top \beta U_q$ is positive is greater than that it is negative.
Accordingly, we see that \Cref{eqn:conditionally_negative} does in fact hold, and the result obtains.
\end{proof}

\section{Gains from \textrm{ECov} in the high-dimensional limit -- supplementary proofs}\label{sec:asymptotic_supp}
\subsection{Proof of \Cref{lemma:gain}}
From the sequence of datasets, $\{\data_D\}_{D=1}^\infty,$ we obtain sequences of estimates.
To make explicit the dimension dependence, we denote these as explicit functions of the data, e.g.\ $\{\betaEC(\data_D)\}_{D=1}^\infty$
where $\betaEC(\data_D)$ denotes $\betaEC$ in \Cref{eqn:empirical_bayes_estimate} applied to $\data_D.$
Furthermore, we consider the entire sequence of datasets and estimates as existing in a single probability space.

We note that \Cref{lemma:mle_to_moments} establishes that
$\betaEC(\data_D)$ and $\betaMM(\data_D)$ coincide almost surely in the high-dimensional limit.
As such, the squared error loss of these two estimates coincide almost surely in the limit, and we may write
\(
\lim_{D\rightarrow \infty} D\inv \risk^D_\pi(\betaEC(\data_D))  
&= \lim_{D\rightarrow \infty} D\inv \E\left[\E[\|\betaEC(\data_D) - \beta\|_F^2 \mid \beta] \right] \\ 
&= \lim_{D\rightarrow\infty}  D\inv \E\big[\E[\|\betaMM(\data_D) - \beta\|_F^2 +\|\betaEC(\data_D) -\betaMM(\data_D)\|_F^2+\\
&2\trace((\betaEC(\data_D) -\betaMM(\data_D))^\top (\betaMM(\data_D) - \beta)) \mid \beta] \big]  \\
&= \lim_{D\rightarrow\infty}\E\left[  D\inv\E[ 
     \|\betaMM(\data_D) - \beta\|_F^2 \mid \beta]\right] \\
&= \lim_{D\rightarrow\infty}\E\left[ \noise Q - \sigma^4 (D-2Q-2)\E[\|\betaLS(\data_D)^\dagger\|_F^2|\beta]\right] \\
&= \noise Q -  \sigma^4\lim_{D\rightarrow\infty}\E[(D-2Q-2)\|\betaLS(\data_D)^\dagger\|_F^2] \\
&= \noise Q -  \sigma^4\lim_{D\rightarrow\infty}\E[\trace[ (\SBTrue+ \noise I_Q)\inv] + o(1)] \\
&= \noise Q - \sigma^4 \trace[ (\SBTrue+\noise  I_Q)\inv].
\)
The third line comes from linearity of expectation and that $\|\betaEC -\betaMM\|\overset{a.s.}{\rightarrow} 0.$
The fourth line comes from \Cref{lemma:form_of_risks}.
The second to last line comes from \Cref{lemma:frob_norm_limit}.

We next recognize that $\trace[ (\SBTrue+\noise  I_Q)\inv] = \sum_{q=1}^Q (\lambda_q +\noise )\inv,$
where $\lambda_1, \dots, \lambda_Q$ are the eigenvalues of $\SBTrue.$
Accordingly we may write,
\(
\lim_{D\rightarrow \infty} D\inv \risk^D_\pi(\betaEC(\data_D))  
&= \noise Q -\sigma^4  \sum_{q=1}^Q (\lambda_q + \noise)\inv.\)

Furthermore since we obtain $\betaID(\data_D)$ by applying $\betaEC(\data_D)$ independently to each dataset,
we analogously obtain
\(
\lim_{D\rightarrow\infty} D\inv \risk^D_\pi(\betaID(\data_D)) 
= \noise Q - \sigma^4 \sum_{q=1}^Q(\SBTrue_{q,q} + \noise)\inv.
\)

Putting these expressions together, we obtain
\(
\lim_{D\rightarrow\infty} D\inv \left[\risk^D_\pi(\betaID(\data_D)) -\risk^D_\pi(\betaEC(\data_D)) \right]
&=\sigma^4 \left[\sum_{q=1}^Q (\lambda_q + \noise)\inv - \sum_{q=1}^Q(\SBTrue_{q,q} +\noise)\inv\right].
\)

Finally, including the additional scaling by $\precision Q\inv$ we obtain
\(
\gain
=\noise Q\inv \left[\sum_{q=1}^Q (\lambda_q + \noise)\inv - \sum_{q=1}^Q(\SBTrue_{q,q} +\noise)\inv\right]
\)
as desired.

\begin{lemma}\label{lemma:mle_to_moments}
Under the conditions of \Cref{lemma:gain}, $\lim_{D\rightarrow\infty} \|\betaEC(\data_D)- \betaMM(\data_D)\|_F = 0$ almost surely.
\end{lemma}
\begin{proof}
Note that under the conditions of \Cref{lemma:gain}, \Cref{lemma:form_of_mle_est} provides that 
    $\betaEC(\data_D)$ and $\betaMM(\data_D)$ differ only when $\SBMM$ is not positive definite;
otherwise $\SBMM=\SBhat.$
    Since $\SBMM=D\inv \betaLS(\data_D)^\top \betaLS(\data_D) -\noise I_Q,$
by \Cref{lemma:strong_convergence_of_sigma} $\SBMM$ will be positive definite for all $D$ above some $D^\prime$ almost surely, 
and so $\betaEC(\data_D)$ and $\betaMM(\data_D)$ become equal for all $D$ large enough, implying strong convergence.
\end{proof}

\begin{lemma}\label{lemma:frob_norm_limit}
    Under the conditions of \Cref{lemma:gain}, $\lim_{D\rightarrow \infty} D \|\betaLS(\data_D)^\dagger\|_F^2  = \trace[ (\SBTrue+\noise I_Q)\inv]$ almost surely.
\end{lemma}
\begin{proof}
Recall that $\|\betaLS(\data_D)^\dagger\|_F^2  = \trace[(\betaLS(\data_D)^\top\betaLS(\data_D))\inv].$
As such, we may write
$D \|\betaLS(\data_D)^\dagger\|_F^2  = \trace[(D\inv \betaLS(\data_D)^\top\betaLS(\data_D))\inv].$
By \Cref{lemma:strong_convergence_of_sigma} $D\inv \betaLS(\data_D)^\top\betaLS(\data_D)\overset{a.s.}{\rightarrow} \SBTrue +\noise I_Q,$
and so we can see that $D \|\betaLS(\data_D)^\dagger\|_F^2  \overset{a.s.}{\rightarrow} \trace[(\SBTrue +\noise I_Q)\inv]$ as desired.
\end{proof}

\begin{lemma}\label{lemma:strong_convergence_of_sigma}
Under the conditions of \Cref{lemma:gain}
$\lim_{D\rightarrow\infty} D\inv \betaLS(\data_D)^\top \betaLS(\data_D) = \SBTrue +\noise I_Q$ almost surely.
\end{lemma}
\begin{proof}
It suffices to show strong convergence element wise, as this implies strong convergence in all other relevant norms.
For convenience, let $C^{(D)} := D\inv \betaLS(\data_D)^\top \betaLS(\data_D).$
Note that we may write each entry $C_{q,q^\prime}^{(D)} = \sum_{d=1}^D D\inv \betaLS(\data_D)^q_d \betaLS(\data_D)_d^{q^\prime}$ as a sum of $D$ i.i.d.\ terms. 
Notably, each term $\betaLS(\data_D)^q_d \cdot \betaLS(\data_D)_d^{q^\prime}$ is a product of two Gaussian random variables and is therefore sub-exponential with some non-negative parameters $(\nu, \alpha)$ (see e.g.\ \citet[Definition 2.7]{wainwright2019high}).
As a result, $C^{(D)}$ is then sub-exponential with parameters $(D\nsqrt \nu, D\inv \alpha).$
Therefore, for any constant $b$ satisfying $0<b<\nu^2 / \alpha,$ by \citet[Proposition 2.9]{wainwright2019high} we have that 
$$
\pr\left[\left|C^{(D)}_{q,q^\prime} - \E[C^{(D)}_{q,q^\prime}]\right|\ge b\right] \le 2\exp\{-\frac{D}{2}b^2/\nu^2\}.
$$
This rapid, exponential decay in tail probability with $D$ implies that for small $b,$
$$
\sum_{D=1}^\infty \pr\left[\left|C^{(D)}_{q,q^\prime} - \E[C^{(D)}_{q,q^\prime}]\right|\ge b\right] \le \infty.
$$
Therefore, by the Borel-Cantelli lemma we see that $|C^{(D)}_{q,q^\prime}-\E[C^{(D)}_{q,q^\prime}]|\overset{a.s.}{\rightarrow} 0.$
Since $\E[C^{(D)}] = \SBTrue +\noise I_Q$ for each $D,$ this implies that 
$\lim_{D\rightarrow\infty} D\inv \betaLS(\data_D)^\top \betaLS(\data_D) = \SBTrue +\noise I_Q$ almost surely.
\end{proof}

\subsection{Further discussion of \Cref{theorem:asymptotic_gain}}\label{sec:asymptotic_gain_supp_discussion}
We here give further detail related to the proof of \Cref{theorem:asymptotic_gain}
and introduce additional notation used in the remainder of the section.
Recall from \Cref{lemma:gain} that 
$\gain =  \noise Q\inv[\sum_{q=1}^Q (\lambda_q + \noise )\inv  - \sum_{q=1}^Q (\SBTrue_{q,q}+\noise )\inv].$
For convenience, we will use $\ell := \diag(\SBTrue)^\downarrow$ to denote the $Q$-vector of diagonal entries of $\SBTrue$
sorted in descending order.
Similarly, we take $\lambda$ to be the $Q$-vector of eigenvalues of $\SBTrue,$ again sorted in descending order.
Next, it is useful to rewrite 
$$
\gain=\noise Q\inv\left[\vec f(\lambda)- \vec f(\ell) \right]
$$
where $\vec f(x):=\sum_{q=1}^Q f(x_q) = \sum_{q=1}^Q (\noise +x_q)\inv$ (where $f(x):=(\noise +x)\inv$).

The key theoretical tool used in establishing \Cref{theorem:asymptotic_gain} is the Schur-Horn theorem.
We state this result below, adapted from \citet[Theorem 5]{horn1954doubly}.
The Schur-Horn theorem guarantees that $\lambda$ majorizes $\ell$.
In particular, an $N$-vector $a$ is said to majorize a second $N$-vector $b$ if $\sum_{n=1}^N a_n = \sum_{n=1}^N b_n$ and 
for all $N^\prime \le N,$
$$\sum_{n=1}^{N^\prime} a_n^\downarrow \ge \sum_{n=1}^{N^\prime} b^\downarrow_n,$$
where for a vector $v,$ we use $v^\downarrow$ to denote the vector with the same components as $v,$ sorted in descending order.
As captured by \Cref{theorem:asymptotic_gain}, we can therefore see that $\gain$ is non-negative for any $\SBTrue$ by observing that $\vec f$ is Schur-convex (since $f$ is convex).

\begin{theorem}[Schur-Horn]\label{theorem:schur_horm}
A vector $\ell$ can be the diagonal of a Hermitian matrix with (repeated) eigenvalues $\lambda$ if and only if $\lambda$ majorizes $\ell.$
\end{theorem}

\subsection{Proof of \Cref{thm:gain_bounds}}

We here show that $\gain$ is upper bounded as
\(
    \gain &\le \noise  Q\inv f^{\prime\prime}(\lambdaMin) \|\lambda\|_2 \|\lambda- \ell  \|_2  \\
                        &= 2 \noise  Q\inv \|\lambda\|_2 \| \lambda-\ell  \|_2/(\noise +\lambdaMin)^3,
\)
and lower bounded as 
\(
    \gain &\ge \frac{1}{2}\noise  Q\inv f^{\prime\prime}(\lambdaMax) \| \lambda-\ell  \|^2  \\
                        &=  \noise Q\inv \| \lambda -\ell  \|^2/(\noise +\lambdaMax)^3,
\)
where $f^{\prime\prime}(x):= \frac{d^2}{dx^2} f(x)$ where $f$ is as defined in \Cref{sec:asymptotic_gain_supp_discussion}.

We obtain both bounds with quadratic approximations to $f$.
In particular, we define $g_\alpha$ as the $2^{\text{nd}}$ order Taylor approximation of $f$ expanded at $\alpha$,
$$g_\alpha(x) := f(\alpha) + f^\prime(\alpha) (x- \alpha) + \frac{1}{2}f^{\prime\prime}(\alpha)(x-\alpha)^2, $$
and note that by \Cref{lemma:quadratic_approx}
\[\label{eqn:quadratic_bounds}
    \vec g_{\lambdaMax}(\lambda)-\vec g_{\lambdaMax}(\ell) 
    \le \vec f(\lambda)-\vec f(\ell) 
    \le \vec g_{\lambdaMin}(\lambda)-\vec g_{\lambdaMin}(\ell) ,
\]
where $\vec g_\alpha(x) := \sum_{q=1}^Q g_\alpha(x_q)$.

\paragraph{Proof of upper bound.}
We obtain the desired upper bound as follows.

\Cref{eqn:quadratic_bounds} and \Cref{lemma:quadratic_gain} allow us to see
\begin{align}\label{eqn:upper_bound_pt1}
\begin{split}
    \gain &\le  \noise  Q\inv \left[\vec g_{\lambdaMin} (\lambda)- \vec g_{\lambdaMin}(\ell) \right]\\
          &= \frac{1}{2}\noise  Q\inv f^{\prime\prime}(\lambdaMin) (\|\lambda\|^2-\| \ell\|^2 ).
\end{split}
\end{align}
Since $f^{\prime\prime}$ is positive on $\R_+,$ the problem reduces to upper bounding $\|\lambda\|^2 - \|\ell\|^2$.

In particular, we find
\begin{alignat}{3}\label{eqn:upper_bound_on_diff_in_norms}
    \|\lambda\|^2 - \|\ell\|^2 &= \langle \lambda + \ell, \lambda-\ell\rangle  &&\\
                                &\le \|\lambda + \ell\| \|\lambda - \ell\| &&\text{// by Cauchy-Schwarz}\\
                                &= \sqrt{\|\lambda\|^2 +2\langle\lambda, \ell\rangle + \|\ell\|^2}\, \|\lambda - \ell\| && \\
&\le \sqrt{\|\lambda\|^2 + 2 \|\lambda \| \|\ell\| + \|\ell\|^2}\, \|\lambda - \ell\| \,\,\,\,
&& \text{// by Cauchy-Schwarz}\\
&\le 2\|\lambda\| \|\lambda - \ell\|
    &&\text{// Since }\|\lambda\| \ge \|\ell\|,
\end{alignat}
where we can see that $\|\lambda\| \ge \|\ell\|$ by noting that $\|\cdot\|^2$ is Schur convex, and again appealing to the Schur-Horn Theorem.
The desired upper bound obtains by combining \Cref{eqn:upper_bound_pt1,eqn:upper_bound_on_diff_in_norms}.

\paragraph{Proof of lower bound.}
We begin as we did for the upper bound.
\Cref{eqn:quadratic_bounds} and \Cref{lemma:quadratic_gain} allow us to see
\begin{align}\label{eqn:lower_bound_pt1}
\begin{split}
    \gain &\ge  \noise  Q\inv \left[\vec g_{\lambdaMax} (\lambda)- \vec g_{\lambdaMax}(\ell)\right] \\
          &= \frac{1}{2}\noise  Q\inv f^{\prime\prime}(\lambdaMax) (\|\lambda\|^2-\| \ell\|^2 ).
\end{split}
\end{align}
Since, again, $f^{\prime\prime}$ is positive on $\R_+,$ the problem reduces to lower bounding $\|\lambda\|^2 - \|\ell\|^2$.

In particular, we would like to show 
$\|\lambda\|^2 - \|\ell\|^2 \ge \|\lambda-\ell\|^2 $.
We can arrive at this bound with a particular expansion of $\|\lambda - \ell\|^2$ and using \Cref{lemma:pos_inner_product},
which again leverages the fact that $\lambda$ majorizes $\ell$.
Specifically, we write
\begin{align}\label{eqn:lower_bound_on_diff_in_norms}
\begin{split}
    \| \lambda -\ell \|^2 &= \langle \lambda - \ell, \lambda\rangle - \langle\lambda -\ell, \ell \rangle \\
                         &= \|\lambda \|^2 - \left[ \langle \lambda, \ell\rangle +  \langle\lambda -\ell, \ell \rangle \right]  \\
                         &= \|\lambda \|^2 - \|\ell\|^2  - \left[ \langle \lambda, \ell\rangle - \langle \ell, \ell \rangle +  \langle\lambda -\ell, \ell \rangle \right] \\
                         &= \|\lambda \|^2 - \|\ell\|^2  - 2\langle \lambda -\ell , \ell\rangle \\ 
                         &\le  \|\lambda \|^2 - \|\ell\|^2 
\end{split}
\end{align}
where the last line follows from \Cref{lemma:pos_inner_product}, which provides that $\langle \lambda -\ell, \ell \rangle\ge 0$ since, from the Schur-Horn theorem for any $Q^\prime \le Q\ \sum_{q=1}^{Q^\prime} \lambda_q - \ell_q \ge 0,$ and $\ell$ has non-negative, non-increasing entries.
We obtain the desired lower bound by combining \Cref{eqn:lower_bound_pt1,eqn:lower_bound_on_diff_in_norms}.

\begin{lemma}\label{lemma:quadratic_approx}
Let $\lambda$ and $\ell$ be $Q$-vectors of non-negative reals with non-increasing entries, and let $\lambda$ majorize $\ell.$
Consider $\vec f: \R^Q \rightarrow \R, x \mapsto \sum_{q=1}^Q f(x_q) = \sum_{q=1}^Q (\noise +x_q)\inv$ (where $f(v):=(\noise +v)\inv$)
for any $\noise>0,$
and define $g_\alpha$ to be the $2^{\text{nd}}$ order Taylor approximation of $f$ expanded at $\alpha$,
$$g_\alpha(x) := f(\alpha) + f^\prime(\alpha) (x- \alpha) + \frac{1}{2}f^{\prime\prime}(\alpha)(x-\alpha)^2.$$
Then 
$$
\vglow(\lambda)-\vglow(\ell) \le \vec f(\lambda)-\vec f(\ell) \le \vghigh(\lambda)-\vghigh(\ell) ,
$$
where $\vec g_\alpha(x) := \sum_{q=1}^Q g_\alpha(x_q)$ and $\lambdaMax=\lambda_1$ and $\lambdaMin=\lambda_Q$ are the largest and smallest entries of $\lambda,$ respectively.
\end{lemma}
\begin{proof}
If there are indices $q$ for which $\lambda_q=\ell_q$, remove them (they do not affect $\vec f(\ell) - \vec f(\lambda)$).
If all are equal, $\lambda=d$ and so the result is trivial, otherwise we have $Q\ge 2$ entries with $\lambda_q \ne \ell_q$.

We begin with the lower bound; the upper bound follows similarly.
For this, it suffices to show
$\vec f(\lambda)- \vec f(\ell)- \left(\vec g_{\lambdaMax}(\lambda)-\vec g_{\lambdaMax}(\ell)\right) \ge 0$.

We first express this difference as an inner product
\(
\vec f(\lambda)-\vec f(\ell)- \left(\vglow(\lambda) - \vglow(\ell)\right) 
     &= \sum_{q=1}^Q \left[(f-\glow)(\lambda_q)-(f-\glow)(\ell_q)\right] \\
     &= \sum_{q=1}^Q (\lambda_q-\ell_q )\left[\frac{(f-\glow)(\lambda_q)-(f-\glow)(\ell_q)}{\lambda_q-\ell_q } \right]\\
     &\text{ // defining each }\, h_q := \frac{(f-\glow)(\lambda_q)-(f-\glow)(\ell_q)}{\lambda_q-\ell_q }\\
     &= \sum_{q=1}^Q (\lambda_q-\ell_q )h_q\\
     &= \langle \lambda-\ell , h \rangle
\)
where $h=[h_1, h_2, \dots, h_Q]^\top.$

We will complete our proof by leveraging \Cref{lemma:pos_inner_product}, which provides that $\langle a, b \rangle \ge 0$ 
for any $Q$-vector $a$ satisfying $\sum_{q=1}^Q a_q=0$ and $\sum_{q=1}^{Q^\prime} a_q \ge 0$ for every $Q^\prime \le Q,$
and $Q$-vector $b$ with non-increasing entries.

It therefore remains only to show that $\lambda - \ell$ and $h$ satisfy the conditions of \Cref{lemma:pos_inner_product}.
Since the entries of $\lambda$ and $\ell$ are taken to be in descending order, 
the condition that $\sum_{q=1}^{Q^\prime} (\lambda - \ell)_q \ge 0$ for any $Q^\prime \le Q,$ follows from the Schur-Horn theorem.
Likewise, this theorem provides that $\sum_{q=1}^Q \lambda_q=\sum_{q=1}^Q \ell_q,$ and therefore that 
$\sum_{q=1}^Q (\lambda- \ell)_q=0,$
so that $\lambda- \ell$ meets condition (2) of the lemma.

We next confirm that $h$ has non-increasing entries by considering an expansion of the expressions for each $h_q$.
In particular, observe that
\(
    h_q &= \frac{(f-\glow)(\lambda_q)-(f-\glow)(\ell_q)}{\lambda_q-\ell_q }\\
        &= (\lambda_q-\ell_q )\inv\left\{f(\lambda_q)- f(\ell_q) - \left[\glow(\lambda_q)- \glow(\ell_q)\right]\right\}\\
        &= (\lambda_q-\ell_q )\inv\big\{\frac{(\noise+\ell_q)-(\noise+\lambda_q) }{(\noise+\ell_q)(\noise+\lambda_q)} - \\
        & \left[(\lambda_q-\ell_q ) f^\prime(\lambdaMax) + \frac{1}{2}  ((\lambda_q-\lambdaMax )^2 -(\ell_q-\lambdaMax )^2)f^{\prime\prime}(\lambdaMax) \right]\big\}\\
        &=  (\noise+\lambdaMax)^{-2} - (\noise+\ell_q)\inv(\noise+\lambda_q)\inv - \frac{1}{2}(\lambda_q-\ell_q )\inv (\noise+\lambdaMax)^{-3} \left[\lambda_q^2-\ell_q^2  - 2 \lambdaMax(\lambda_q-\ell_q )\right]\\
        &=  (\noise+\lambdaMax)^{-2} - (\noise+\ell_q)\inv(\noise+\lambda_q)\inv - \frac{1}{2}(\noise+\lambdaMax)^{-3} \left[\lambda_q+ \ell_q  - 2 \lambdaMax\right].
\)

Next define $\phi(a,b)= (\noise+\lambdaMax)^{-2} - (\noise+a)\inv(\noise+b)\inv - \frac{1}{2}(\noise+\lambdaMax)^{-3} \left[b+ a- 2 \lambdaMax\right],$
so that for each $q, h_q = \phi(\ell_q,\lambda_q).$
Now, for $q^\prime > q,$ we may write
\begin{align}\label{eqn:hq_diff_by_partials}
\begin{split}
h_{q^\prime} - h_{q} 
&= \phi(\ell_{q^\prime}, \lambda_{q^\prime}) -  \phi(\ell_{q}, \lambda_{q})\\
&= \int_{\ell_q}^{\ell_{q^\prime}} \frac{\partial}{\partial a}\phi(a, \lambda_q) da+ 
\int_{\lambda_q}^{\lambda_{q^\prime}} \frac{\partial}{\partial b}\phi(\ell_{q^\prime}, b)db.
\end{split}
\end{align}
Next note that 
$$\frac{\partial}{\partial a}\phi(a, b) = (\noise+a)^{-2}(\noise+b)\inv  - \frac{1}{2}(\noise+\lambdaMax)^{-3}$$
 and
$$\frac{\partial}{\partial b}\phi(a, b) = (\noise+a)\inv(\noise+b)^{-2}  - \frac{1}{2}(\noise+\lambdaMax)^{-3}$$
from which we can see that 
$\frac{\partial}{\partial a}\phi(a, b) $ and $\frac{\partial}{\partial b}\phi(a, b) $ 
are positive for $a, b \in [\lambdaMin, \lambdaMax].$
Accordingly, \Cref{eqn:hq_diff_by_partials} provides that $h_{q^\prime}-h_q \le 0,$
since $\ell_{q^\prime} \le \ell_q$ and $\lambda_{q^\prime} \le \lambda_q$ for $q^\prime >q,$
because the entries of $\ell$ and $\lambda$ are non-increasing.
Therefore $h_{q^\prime} \le h_q,$ completing the proof.
\end{proof}

\begin{lemma}\label{lemma:quadratic_gain}
Consider the quadratic function $\vec h(x) = \sum_{q=1}^Q (ax_q^2 + bx_q +c)$.
Let $\lambda, \ell \in \R^Q$ satisfy $\sum_{q=1}^Q \lambda_q = \sum_{q=1}^Q \ell_q$.
Then
$$
\vec h(\ell) - \vec h(\lambda) = a(\|\ell\|^2 - \|\lambda\|^2).
$$
\end{lemma}
\begin{proof}
The result follows from the simple algebraic rearrangement below,
\(
\vec h(\ell) - \vec h(\lambda) &= \sum_{q=1}^Q  (a\ell_q^2 + b\ell_q +c) - (a\lambda_q^2 + b \lambda_q + c)\\
&= \sum_{q=1}^Q  a\ell_q^2 - a\lambda_q^2\\
&= a (\|\ell\|^2 - \|\lambda\|^2).
\)
\end{proof}

\begin{lemma}\label{lemma:pos_inner_product}
Let $x$ be a $Q$-vector satisfying for each $Q^\prime \le Q, \, \sum_{q=1}^{Q^\prime}x_q \ge 0$, 
and let $y$ be a $Q$-vector with non-increasing entries.
If additionally either 
(1) $y$ has non-negative entries or
(2) $\sum_{q=1}^Q x_q=0$
then  
$\langle x, y \rangle \ge y_Q \sum_{q=1}^Q x_q \ge 0.$
%Additionally, the condition holds for positive, non-decreasing $y$ and any $x$ which satisfies the original condition once the order of its elements are reversed.
\end{lemma}
\begin{proof}
We first prove the lemma under condition (1) by induction.
The base case of $Q=1$ is trivial; $\langle x, y \rangle = x_1 y_1$ and under (1) $x_1$ and $y_1$ are non-negative and under (2) $x_1=0.$

Assume the result holds for $Q-1$.
Then 
\begin{alignat}{3}
    \langle x, y \rangle &=y_Q x_Q +  \langle x_{1:Q-1}, y_{1:Q-1} \rangle  &&\\
                     &\ge y_Q x_Q +  y_{Q-1}\sum_{q=1}^{Q-1}x_q 
                     && \text{// by the inductive hypothesis} \\ 
&\ge y_Q x_Q +  y_{Q}\sum_{q=1}^{Q-1}x_q  \,\,\,\, 
&& \text{// since }y_{Q-1} \ge y_Q \text{ and } \sum_{q=1}^{Q-1}x_q\ge 0 \\
&= y_Q \sum_{q=1}^{Q}x_q \ge 0 
    && \text{// since }y_Q \text{ and } \sum_{q=1}^{Q}x_q \text{ are non-negative.}
\end{alignat}
This provides the desired inductive step, completing the proof under condition (1).

Under condition (2), consider $y^\prime = y - \min_{q} y_q \ones{Q}.$
Then
\(
\langle x, y \rangle 
&= \langle x, y^\prime \rangle  +  \min_{q} y_q   \langle x, \ones{Q} \rangle \\
&= \langle x, y^\prime \rangle. 
\)
Since $y^\prime$ now has non-negative entries, condition (1) is satisfied and the result follows.
\end{proof}

\subsection{Proof of \Cref{cor:gain_SNR}}
We establish the corollary with a brief sequence of upper bounds following from our initial upper bound in 
\Cref{theorem:asymptotic_gain}.
In particular, the theorem provides
$$\gain \le  2\noise Q\inv \|\lambda^\downarrow\| \| \ell^\downarrow - \lambda^\downarrow \|/(\noise + \lambdaMin)^3.$$

We begin by simplifying this upper bound.
As a first step, note that 
\( 
\| \ell^\downarrow - \lambda^\downarrow \|^2 
&= \|\ell\|^2 + \|\lambda\|^2 - \langle \ell^\downarrow, \lambda^\downarrow\rangle\\
&\le 2\|\lambda\|^2 .
\)

As such, we can simplify our upper bound as
\begin{align}\label{eqn:general_SNR_bound}
\begin{split}
\gain &\le  2\noise Q\inv \|\lambda\|\|\ell^\downarrow - \lambda^\downarrow \| /(\noise+ \lambdaMin)^3\\
&\le  4\noise  Q\inv \|\lambda\|^2 /(\noise+ \lambdaMin)^3\\
&\le  4\kappa^2\lambdaMin^2 \noise /(\noise+ \lambdaMin)^3
\end{split}
\end{align}
where $\kappa := \lambdaMax/\lambdaMin$ is the condition number of $\SBTrue.$

We then obtain the first bound by noting that 
\(
\lambdaMin^2 \noise /(\noise+ \lambdaMin)^3
&\le \lambdaMin^2 \noise /(\noise)^2/\lambdaMin \\
&\le \lambdaMin /\noise
\)
and the second by noting that 
\(
\lambdaMin^2 \noise /(\noise+ \lambdaMin)^3
&\le \lambdaMin^2 \noise /(\lambdaMin)^3 \\
&\le \noise /\lambdaMin.
\)
Substituting these expressions into \Cref{eqn:general_SNR_bound} provides the desired expressions in \Cref{cor:gain_SNR}.

\section{Experiments Supplementary Results and Details}\label{sec:experiments_supp}

\subsection{Simulations additional details}
We here describe the details of the simulated datasets discussed in \Cref{sec:experiments}.
For each of the dimensions $D$ and each of the 20 replicates we first generated covariate effects for all $Q=10$ datasets.
To do this, we began by setting $\SB;$
for the correlated covariate effects experiments (\Cref{fig:simulation} Left) we 
generating a random $Q\times Q$ matrix of orthonormal vectors $U$ and set $\SB = U\diag([2^0, 2\inv, \dots, 2^{Q-1}]^\top) U^\top,$
and for independent effects (\Cref{fig:simulation} Right) we set $\SB = I_Q.$
We then simulated covariate effects as $\beta_d\overset{i.i.d.}{\sim}\mathcal{N}(0, \SB).$

We next simulated the design matrices.
For each dataset $q,$ we chose a random number of data points $N^q\sim\text{Pois}(\lambda=1000)$, 
and for each data point $n=1, \dots, N^q$ sampled $X^q_n \sim \mathcal{N}(0, (1/1000) I_D)$ 
so that for each dataset $\E[X^{q\top}X^q]= I_D.$
Finally, we generated each response as $Y^q_n \overset{indep}{\sim}\mathcal{N}(X^{q\top}_n \beta^q, 1).$

For $\betaED,$ we estimated the $D\times D$ covariance $\Gamma$ by maximum marginal likelihood.
We did this with an EM algorithm closely related to \Cref{alg:EM_general}.
See e.g.\ \citet[Chapter 15 sections 4-5]{gelman2013bayesian} for an explanation of the relevant conjugacy calculations in a more general case that includes a hyper-prior on $\Gamma.$

\subsection{Practical moment estimation for poorly conditioned problems}\label{sec:real_data_moments}
The moment based estimator (using $\SBMM$ in \Cref{sec:nonasymptotic_theory}) is unstable in the two real data applications discussed in \Cref{sec:experiments} due to poor conditioning of the design matrices leading $\betaLS$ to have high variance.
To overcome this limitation, we instead used an adapted moment estimation procedure which is less sensitive to this poor conditioning.
While, in agreement with \Cref{theorem:pos_part_dominance}, this approach performs worse than $\betaEC$ (see \Cref{fig:law_and_blog_aggregate_supp})
we report it nonetheless because it has lower computational cost and may be appealing for larger scale applications.
We describe this approach here.
We note however that moment based estimates of the sort we consider here do not naturally extend to logistic regression
and so are not reported for our application to CIFAR10.

We first introduce some additional notation.
For each dataset $q$ consider the reduced singular value decomposition $X^q = S^q\diag(\omega^q)R^{q\top},$
where $S^q$ and $R^q$ are $N^q\times D$ and $D\times D$ matrices with orthonormal columns and $\omega^q$ is a $D$-vector of non-negative singular values.
Next define for each dataset $W^q := S^{q\top} X^q$ and $Z^q := S^{q\top}Y^q,$ 
which we may interpret as a $D\times D$ matrix of pseudo-covariates and $D$-vector of pseudo-responses, respectively.
Next define $\Omega$ to be the $Q\times Q$ matrix with entries $\Omega_{q,q^\prime}:= \trace(W^{q\top}W^{q^\prime})\inv$
and $\vec \sigma^2 :=[\sigma_1^2, \sigma_2^2,\dots, \sigma_Q^2]^\top.$
Lastly, let $Z = [Z^1, Z^2, \dots, Z^Q]$ be the $D\times Q$ matrix of all pseudo-responses.
Our new moment estimator is 
$$
\SBMM := [Z^\top Z - D\diag(\vec \sigma^2) ]\odot \Omega.
$$

We next show hat $\E[\SBMM]=\SB$ under correct prior and likelihood specification.
Note first that if $\delta$ is a $D\times Q$ matrix with i.i.d.\ standard normal entries we may write
$$Z \overset{d}{=} [W^1\beta^1, W^2\beta^2, \dots, W^Q\beta^Q]  + \delta\diag(\vec \sigma^2).$$
As such, for each $q$ and $q^\prime,$ we have that 
\(
\E[(Z^\top Z)_{q,q^\prime}] 
&=\E[Z^{q\top} Z^q]\\
&=\E[\beta^{q\top}W^{q\top}W^{q^\prime}\beta^{q^\prime}] + \mathbb{I}[q=q^\prime]\sigma_q^2 D\\
&=\trace(W^{q\top}W^{q^\prime}\E[\beta^{q^\prime}\beta^{q\top}]) + \mathbb{I}[q=q^\prime]\sigma_q^2 D\\
&=\Omega_{q,q^\prime}\inv \SB_{q,q^\prime} + \mathbb{I}[q=q^\prime]\sigma_q^2 D.
\)
Accordingly, we can see that each entry of $\SBMM$ has expectation
$\E[\SBMM_{q,q^\prime}] = \SB_{q,q^\prime},$
which establishes unbiasedness.

However, this moment estimate still has the limitation that it evaluates to a non positive semidefinite matrix with positive probability.
Under the expectation that, in line with \Cref{theorem:exch_reg_domination} the very small and negative eigenvalues of $\SBMM$ might lead to over-shrinking, 
we performed an additional step of clipping these eigenvalues to force the resulting estimate to be reasonably well conditioned.
In particular, if our initial estimate had eigendecomposition $\SBMM=U\diag(\lambda)U^\top,$ 
we instead used $\SBMM = U\diag(\tilde \lambda)U^\top,$
where for each $q,$ we have $\tilde \lambda_q = \max(\lambda_q, \lambdaMax/100)$
so that the condition number of the modified estimate was at most 100.
Though we did not find the performance of the resulting estimates to be very sensitive to this cutoff,
we view requirement for these partly subjective implementation choices required to make the $\betaMM$
effective in practice to be a downside of the approach as compared to $\betaEC,$ 
which avoids such choices by estimating $\SB$ by maximum marginal likelihood.

Compared to the iterative EM algorithms, which rely on matrix inversions at each iteration,
computation of $\SBMM$ is much faster.
In each of our experiments, computing it requires less than one second.

\begin{figure}
    \centering
    \includegraphics[width=\textwidth]{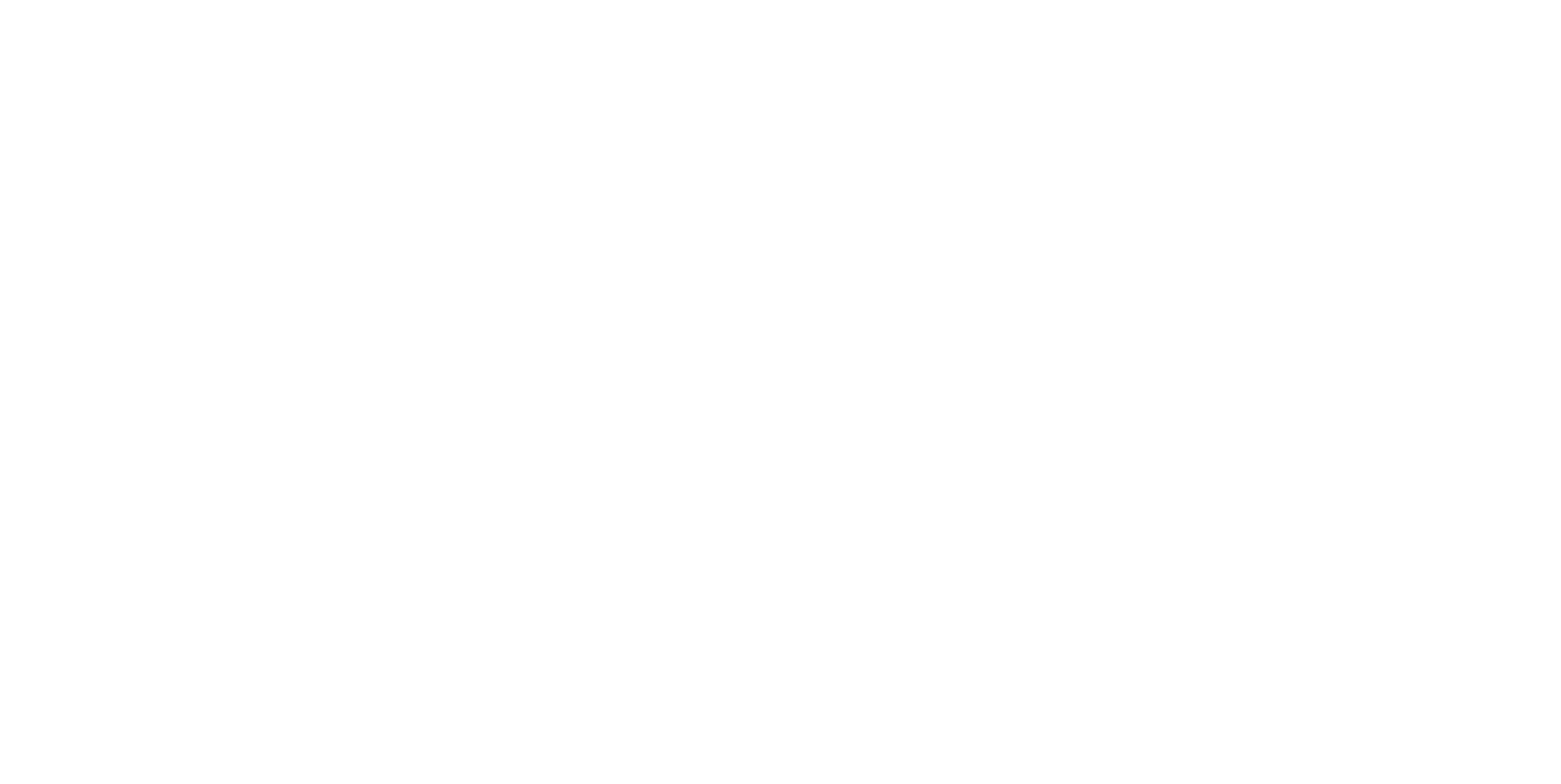}
    \caption{Performances of additional methods on the law enforcement and blog datasets.  Uncertainty intervals are $\pm 1\texttt{SEM}.$}
    \label{fig:law_and_blog_aggregate_supp}
\end{figure}

\subsection{Allowing for non-zero means a priori in hierarchical Bayesian estimates}\label{sec:non_zero_means}
In the development of our approach in \Cref{sec:exch_cov} we imposed the restriction that
$\E[\beta_d]=0$ a priori.
Though in general one might prefer to let $\beta$ have some nontrivial mean
(as \citet{lindley1972bayes} do in the context of exchangeability of effects across datasets) this assumption simplifies the resulting estimators, theory, and notation.
When $\beta$ is permitted to have a non-zero mean, conjugacy maintains and the methodology presented in \Cref{sec:our_method} may be updated to 
accommodate the change.
While we omit a full explanation of the tedious details of this variation,
we include its implementation in our code and the performance of the resulting empirical Bayesian estimators in \Cref{fig:law_and_blog_aggregate_supp,fig:blog_segmented_supp,fig:law_segmented_supp}.
From these empirical results we see that removing this restriction has little impact on the performance of the resulting estimators.
Notably, our results in these figures reveal that the same is true for choosing to include or exclude a prior mean for the exchangeability of effects across datasets prior.

\begin{figure}
    \centering
    \includegraphics[width=\textwidth]{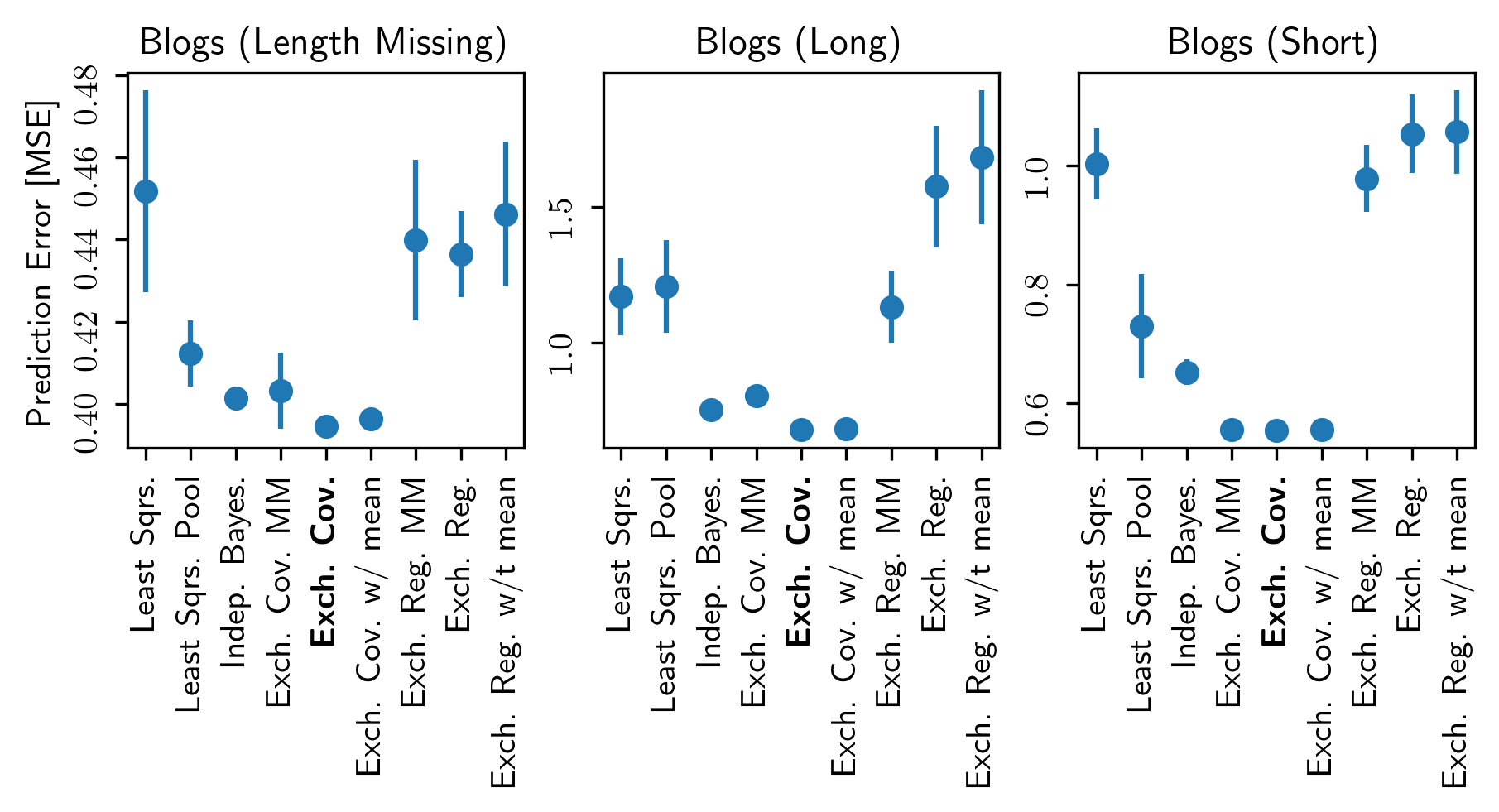}
    \caption{Performances of methods on the blog dataset, segmented by post type.  Uncertainty intervals are $\pm 1\texttt{SEM}.$}
    \label{fig:blog_segmented_supp}
\end{figure}

\begin{figure}
    \centering
    \includegraphics[width=1.0\textwidth]{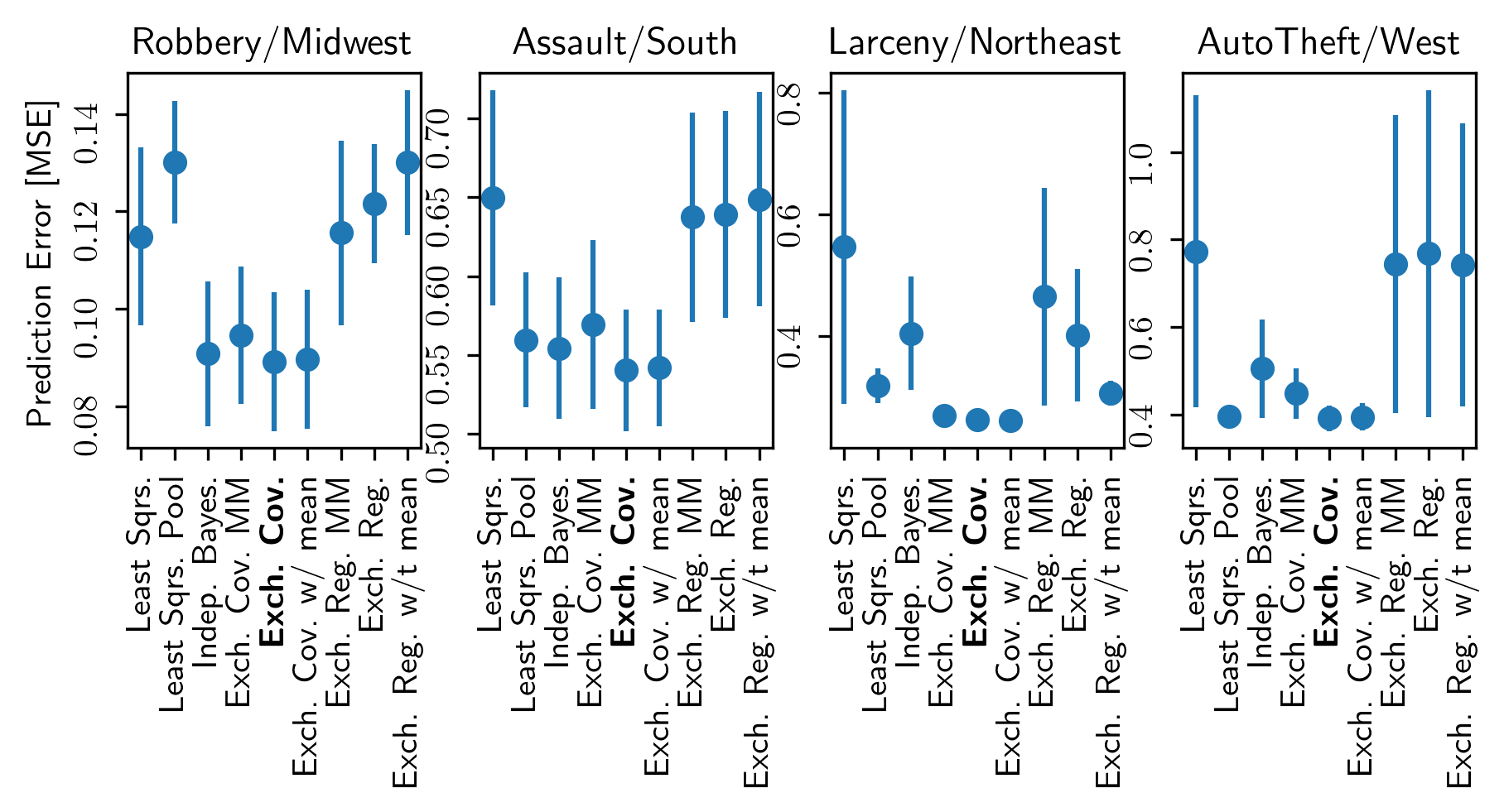}
    \caption{Performances of methods on the law enforcement dataset, segmented by region and recorded offense categorization.  Uncertainty intervals are $\pm 1\texttt{SEM}.$}
    \label{fig:law_segmented_supp}
\end{figure}

\subsection{Additional details on datasets} \label{sec:real_data_supp}
In each of the two regression applications, for each component dataset, we mean centered and variance-normalized the responses.
Additionally, we Winsorized the responses by dataset;
in particular, we clipped values more than 2 standard deviations from the mean.

\paragraph{BlogFeedback Data Set details}
Given the nature of the features included in the blog dataset used in the main text (which are summarizing characteristics rather than readable text), we believe it may be possible to find the blog post that corresponds to a particular data point. But we believe it is unlikely that the dataset directly contains any personally identifiable information. The blog information was obtained by web-crawling on publicly posted pages, so it is unlikely that consent for inclusion of the content into this dataset was obtained.

\paragraph{Communities and Crime Dataset details} 
All data in this dataset was obtained through official channels. 
This dataset is composed of statistics aggregated at the community level, so it is less likely (though not impossible) to contain personally identifiable information. Since it contains demographic, census, and crime data, it is unlikely to contain offensive content.

\paragraph{CIFAR10 details.}
For the tasks
\texttt{car vs.\ cat},
\texttt{car vs.\ dog},
\texttt{truck vs.\ cat}, and
\texttt{truck vs.\ dog} we used $N^q=100$ data points.
For the tasks
\texttt{car vs.\ deer},
\texttt{car vs.\ horse},
\texttt{truck vs.\ deer}, and
\texttt{truck vs.\ horse} we used $N^q=1000$ data points.

We generated the pre-trained neural network embeddings using a variational auto-encoder (VAE) \citep{kingma2013auto}.
We adapted our VAE implementation from \texttt{ALIBI DETECT} \citep{alibi2019detect},
\href{https://github.com/SeldonIO/alibi-detect/blob/7decb07d6739a0ede47b104bb746ead090822145/examples/od_vae_cifar10.ipynb}{here}.
See also \texttt{notebooks/2021\_05\_12\_CIFAR10\_VAE\_embeddings.ipynb} for details.

CIFAR10 is composed from a subset of the 80 million tiny images dataset.
As is currently acknowledged on the 80 million tiny images website, this larger dataset is known to contain offensive images and images obtained without consent (\url{https://groups.csail.mit.edu/vision/TinyImages/}).
However, given the benign nature of the 10 image classes in CIFAR10, we expect it does not contain offensive or personally identifiable content.
These data were also obtained by web-crawling, so it is unlikely that consent for inclusion of the content into this dataset was obtained.

\begin{figure}
    \centering
    \includegraphics[width=1.0\textwidth]{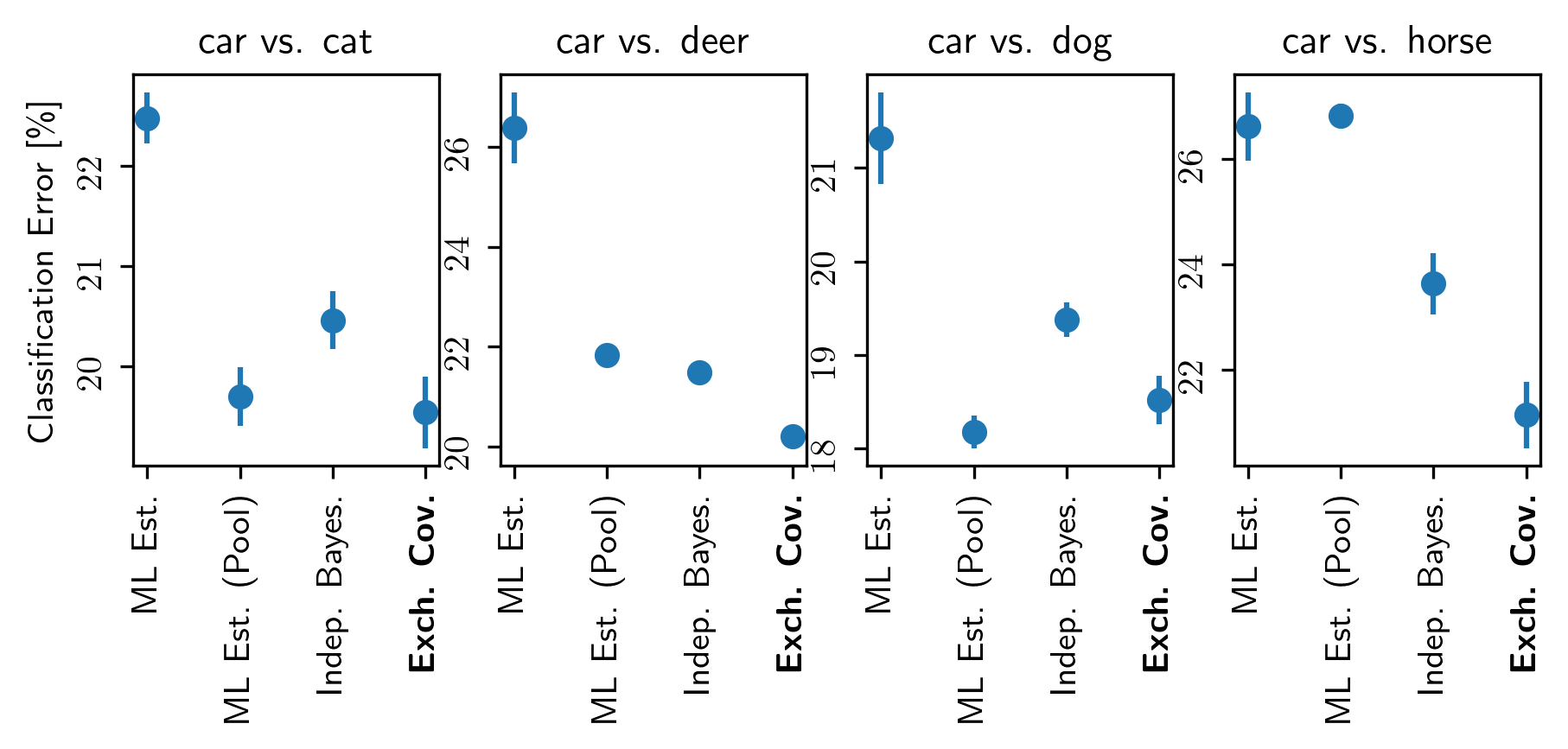}
    \includegraphics[width=1.0\textwidth]{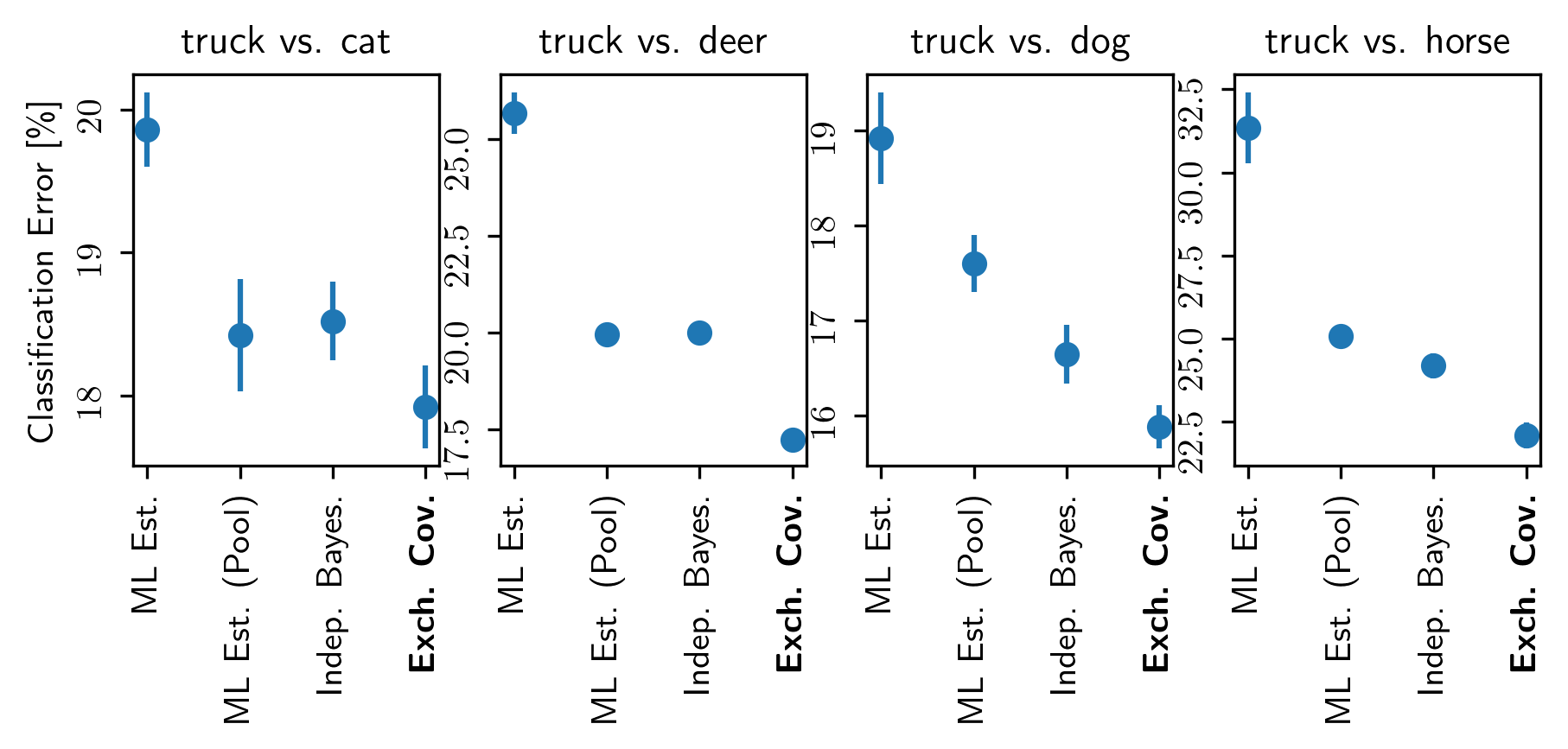}
    \caption{Performances of methods on CIFAR10 segmented by binary classification task.  Uncertainty intervals are $\pm 1\texttt{SEM}.$}
    \label{fig:CIFAR_segmented_supp}
\end{figure}

\subsection{Software Licenses}\label{sec:licenses}
We here report the software used to generate our results and their associated licenses.

All of our experiments were implemented in \texttt{python}, which is licensed under the PSF license.
For ease of reproducibility, ran our experiments and generated our plots IPython in Jupyter notebooks;
this software is covered by a modified BSD license.

For our application to transfer learning using CIFAR10, we used a variational auto-encoder implementation adapted from \texttt{ALIBI DETECT} \citep{alibi2019detect}, which uses the Apache licence.
Our implementation of our EM algorithm uses \texttt{TensorFlow} \citep{abadi2016tensorflow}, which is licensed under the MIT license.

We made frequent use of python packages \texttt{numpy} and \texttt{scipy} and \texttt{matplotlib}.
These are large libraries with components covered different licenses.
See \href{https://github.com/scipy/scipy/blob/master/LICENSES_bundled.txt}{github.com/scipy/scipy/blob/master/LICENSES\_bundled.txt} for \texttt{scipy},
\href{https://github.com/numpy/numpy/blob/main/LICENSES_bundled.txt}{github.com/numpy/numpy/blob/main/LICENSES\_bundled.txt} for \texttt{numpy}, and 
\href{https://github.com/matplotlib/matplotlib/tree/master/LICENSE}{github.com/matplotlib/matplotlib/tree/master/LICENSE} for \texttt{matplotlib}.

\end{document}